\documentclass[preprint,fleqn,12pt]{elsarticle}

\usepackage{amsmath,amssymb,amsfonts}%

\usepackage{lineno}

\usepackage{booktabs}%
\usepackage[T1]{fontenc}

\usepackage{tikz}
\usetikzlibrary{shapes,calc,positioning,arrows}
\usepackage{paralist}
\usepackage[ruled,linesnumbered] 
       {algorithm2e} 
\SetAlgoInsideSkip{medskip}

\usepackage{float}
\usepackage{wrapfig}

\usepackage{hyperref}
\hypersetup{
  colorlinks   = true, 
  urlcolor     = blue, 
  linkcolor    = blue, 
  citecolor   = red 
}

\newtheorem{theorem}{Theorem}
\newtheorem{lemma}{Lemma}%
\newproof{proof}{Proof}%

\newcommand{\Po}{\mathsf{P}}
\newcommand{\NP}{\mathsf{NP}}
\newcommand{\NPC}{\mathsf{NPC}}

\def\MC{{\sc mc}} 
\def\PMC{{\sc pmc}} 
\def\DPM{{\sc dpm}} 

\def\threeSAT{\textup{\textsc{3sat}}}
\def\NAE3SAT{\textup{\textsc{nae 3sat}}}
\def\oneIN3SAT{\textup{\textsc{1-in-3sat}}}
\def\P1IN3SAT{{\sc positive 1-in-3sat}}

\newtheorem{observation}{Observation}
\newtheorem{fact}{Fact}

\definecolor{teal}{rgb}{0.00,0.5,0.5} 
\definecolor{flax}{rgb}{0.90,0.78,0.00} 
\colorlet{hellgrau}{black!12!white}

\begin{document}

\begin{frontmatter}

\title{Complexity and algorithms for matching cut problems in graphs without long induced paths and cycles\tnoteref{wg}}
\tnotetext[wg]{A preliminary version~\cite{LeL23} of this paper appeared in the proceedings of the 49th International Workshop on Graph-Theoretic Concepts in Computer Science (WG~2023). This full version contains new results for matching cut problems in graphs without long induced cycles.}

\author[1]{Hoang-Oanh Le}
\ead{hoangoanhle@outlook.com}

\author[2]{Van Bang Le\corref{cor}} 
\ead{van-bang.le@uni-rostock.de}
\cortext[cor]{Corresponding author}

 \affiliation[1]{organization={Independent Researcher},
            state={Berlin},
            country={Germany}}   
            
\affiliation[2]{organization={Institut f{\"u}r Informatik, Universit{\"a}t Rostock},
            state={Rostock},
            country={Germany}} 

\begin{abstract}
In a graph, a (perfect) matching cut is an edge cut that is a (perfect) matching. \textsc{matching cut} (\MC), respectively, \textsc{perfect matching cut} (\PMC), is the problem of deciding whether  a given graph has a matching cut, respectively, a perfect matching cut. 
The \textsc{disconnected perfect matching} problem (\DPM) is to decide if a graph has a perfect matching that contains a matching cut. 
Solving an open problem posed in [Lucke, Paulusma, Ries (ISAAC 2022, Algorithmica 2023)], we show that \PMC\ is $\NP$-complete in graphs without induced $14$-vertex path $P_{14}$. Our reduction also works simultaneously for \MC\ and \DPM, improving the previous hardness results of \MC\ on $P_{15}$-free graphs and of \DPM\ on $P_{19}$-free graphs to $P_{14}$-free graphs for both problems.
Actually, we prove a slightly stronger result: within $P_{14}$-free $8$-chordal graphs (graphs without chordless cycles of length at least~$9$), it is hard to distinguish between those without matching cuts (respectively, perfect matching cuts, disconnected perfect matchings) and those in which every matching cut is a perfect matching cut. 
Moreover, assuming the Exponential Time Hypothesis, none of these problems can be solved in $2^{o(n)}$ time for $n$-vertex $P_{14}$-free $8$-chordal graphs.  

On the positive side, we show that, as for \MC\ [Moshi (JGT 1989)], \DPM\ and \PMC\ are polynomially solvable when restricted to $4$-chordal graphs. Together with the negative results, this partly answers an open question on the complexity of \PMC\ in $k$-chordal graphs asked in [Le, Telle (WG 2021, TCS 2022) \& Lucke, Paulusma, Ries (MFCS 2023, TCS 2024)].

\end{abstract}


\end{frontmatter}


\section{Introduction and results}
In a graph $G=(V,E)$, a \emph{cut} is a partition $V=X\cup Y$ of the vertex set into disjoint, non-empty sets $X$ and $Y$. The set of all edges in $G$ having an endvertex in~$X$ and the other endvertex in~$Y$, written $E(X,Y)$, is called the \emph{edge cut} of the cut $(X,Y)$. A \emph{matching cut} is an edge cut that is a (possibly empty) matching. 
Another way to define matching cuts is as follows; see~\cite{Chvatal84,Graham70}: a cut $(X,Y)$ is a matching cut if and only if each vertex in~$X$ has at most one neighbor in~$Y$ and each vertex in~$Y$ has at most one neighbor in~$X$. 
The classical $\NP$-complete problem \textsc{matching cut} (\MC)~\cite{Chvatal84} asks if a given graph admits a matching cut.  

An interesting special case, where the edge cut $E(X,Y)$ is a \emph{perfect matching} of~$G$,  
was considered in~\cite{HeggernesT98}. Such a matching cut is called a \emph{perfect mathing cut} and the \textsc{perfect matching cut} (\PMC) problem asks whether a given graph admits a perfect matching cut. It was shown in~\cite{HeggernesT98} that this special case \PMC\ of \MC\ remains $\NP$-complete.

A notion related to matching cut is \emph{disconnected perfect matching} which has been considered recently in~\cite{BouquetP25}: a disconnected perfect matching is a perfect matching that contains a matching cut. Observe that any perfect matching cut is a disconnected perfect matching but not the converse. Fig.~\ref{fig:examples} provides some small examples for matching cuts, perfect matching cuts and disconnected perfect matchings.

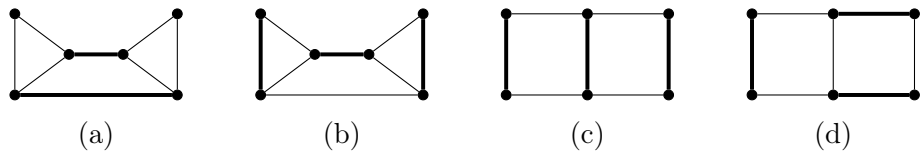
\begin{figure}[!ht]
\begin{center}
\tikzstyle{vertexS}=[draw,circle,inner sep=2pt,fill=black] 
\tikzstyle{vertex}=[draw,circle,inner sep=1.3pt,fill=black] 

\begin{tikzpicture}[scale=.36] 
\node[vertex] (a) at (1,1)  {};
\node[vertex] (b) at (1,4)  {};
\node[vertex] (c) at (7,4)  {};  
\node[vertex] (d) at (7,1)  {};
\node[vertex] (e) at (3,2.5)  {};
\node[vertex] (f) at (5,2.5)  {}; 

\draw (a)--(b)--(e); \draw[line width=1.5pt] (e)--(f); \draw (f)--(c)--(d); \draw[line width=1.5pt] (d)--(a); \draw (a)--(e);
\draw (f)--(d); 

\node at (4,-.5) {\small (a)};
\end{tikzpicture} 
\qquad
\begin{tikzpicture}[scale=.36] 
\node[vertex] (a) at (1,1)  {};
\node[vertex] (b) at (1,4)  {};
\node[vertex] (c) at (7,4)  {};  
\node[vertex] (d) at (7,1)  {};
\node[vertex] (e) at (3,2.5)  {};
\node[vertex] (f) at (5,2.5)  {}; 

\draw[line width=1.5pt] (a)--(b); \draw (b)--(e); \draw[line width=1.5pt] (e)--(f); \draw[line width=1.5pt] (d)--(c);
\draw (f)--(c); \draw (d)--(a)--(e);
\draw (f)--(d); 

\node at (4,-.5) {\small (b)};
\end{tikzpicture} 
\qquad
\begin{tikzpicture}[scale=.36] 
\node[vertex] (a) at (1,1)  {};
\node[vertex] (b) at (1,4)  {};
\node[vertex] (c) at (4,4)  {};  
\node[vertex] (d) at (4,1)  {};
\node[vertex] (e) at (7,1)  {};
\node[vertex] (f) at (7,4)  {}; 

\draw[line width=1.5pt] (a)--(b); \draw (b)--(c); \draw[line width=1.5pt] (d)--(c);
\draw (d)--(a);
\draw (d)--(e); \draw[line width=1.5pt] (e)--(f); \draw (f)--(c); 

\node at (4,-.5) {\small (c)};
\end{tikzpicture} 
\qquad
\begin{tikzpicture}[scale=.36] 
\node[vertex] (a) at (1,1)  {};
\node[vertex] (b) at (1,4)  {};
\node[vertex] (c) at (4,4)  {};  
\node[vertex] (d) at (4,1)  {};
\node[vertex] (e) at (7,1)  {};
\node[vertex] (f) at (7,4)  {}; 

\draw[line width=1.5pt] (a)--(b); \draw (b)--(c)--(d)--(a);
\draw[line width=1.5pt] (d)--(e); \draw (e)--(f); \draw[line width=1.5pt] (c)--(f); 

\node at (4,-.5) {\small (d)};
\end{tikzpicture} 
\caption{Some example graphs; bold edges indicate a matching in question. (a):  a matching cut. (b): a perfect matching that is neither a matching cut nor a disconnected perfect matching; this graph has no disconnected perfect matching, hence no perfect matching cut. (c): a perfect matching cut, hence a disconnected perfect matching. (d): a disconnected perfect matching that is not a perfect matching cut.}\label{fig:examples}
\end{center}
\end{figure}

The related problem to \MC\ and \PMC, \textsc{disconnected perfect matching} (\DPM), asks if a given graph has a disconnected perfect matching; equivalently: if a given graph has a matching cut that is extendable to a perfect matching. It was shown in~\cite{BouquetP25} that \DPM\ is $\NP$-complete. All these three problems have received much attention lately; see, e.g., ~\cite{BonnetCD24,BouquetP25,ChenHLLP21,FeghaliLPR25,GolovachKKL22,LeT22,Lucke25,LuckePR22,LuckePR23,LuckePR24} for recent results. 
In particular, we refer to~\cite{LuckePR23} for a short survey on known applications of matching cuts in other research areas.

In this paper, we focus on the complexity and algorithms of these three problems restricted to graphs without long induced paths and cycles.  
$P_4$-free graphs have already been considered by Bonsma~\cite{Bonsma09}. He proved that the $4$-cycle is the only connected $P_4$-free graph with minimum degree at least two that has a matching cut. 
In~\cite{Feghali23},  Feghali initiated the study of \MC\ for graphs without a fixed path as an induced subgraph. He showed that \MC\ is solvable in polynomial time for $P_5$-free graphs, but that there exists an integer $t>0$ for which it is $\NP$-complete for $P_t$-free graphs; indeed, $t=27$ as pointed out in~\cite{LuckePR23}. 
The current best known hardness results for \MC\ and \DPM\ in graphs without long induced paths are:
\begin{theorem}[\cite{LuckePR23}]\label{thm:P19mc}
\MC\ remains $\NP$-complete in $\{3P_5,P_{15}\}$-free graphs, \DPM\ remains $\NP$-complete in $\{3P_7,P_{19}\}$-free graphs.
\end{theorem}
Prior to the present paper, no similar hardness result for \PMC\ was known. 
Indeed, it was asked in~\cite{LuckePR23}, whether there is an integer $t$ such that \PMC\ is $\NP$-complete in $P_t$-free graphs. 
Polynomial-time algorithms exist for \MC,  \PMC\ and \DPM\ in $P_6$-free graphs~\cite{LuckePR23}. 

Graphs without induced cycles of length larger than $k$ are \emph{$k$-chordal graphs}. The \emph{chordality} of a graph $G$ is the smallest integer $k$ such that $G$ is $k$-chordal. Many classical graph classes consist of $k$-chordal graphs for small $k$: chordal graphs are $3$-chordal, chordal bipartite graphs, weakly chordal graphs and cocomparability graphs are $4$-chordal. Observe that $k$-chordal graphs are particularly $(k+1)$-chordal. 
For graphs without long induced cycles, the only result we are aware of is that \MC\ is polynomially solvable in $4$-chordal graphs: 
\begin{theorem}[\cite{Moshi89}]\label{thm:longholefree-mc}
There is a polynomial-time algorithm solving \MC\ in $4$-chordal graphs. 
\end{theorem}

Previously, no similar polynomial-time results for \DPM\ and \PMC\ in $4$-chordal  were known. 
Indeed, the computational complexity of \PMC\ in $k$-chordal graphs, $k\ge 4$, has been stated as an open problem in~\cite{LeT22,LuckePR23}.  

\paragraph{Our contributions}
We prove that \PMC\ is $\NP$-complete in graphs without induced path $P_{14}$, solving the open problem posed in~\cite{LuckePR23}. For \MC\ and \DPM\, we improve the hardness results in Theorem~\ref{thm:P19mc} in graphs without induced path $P_{15}$, respectively, $P_{19}$, to graphs without induced path $P_{14}$. It is remarkable that all these hardness results for \emph{three} problems will be obtained simultaneously by only \emph{one} reduction. Thus, our approach unifies and improves known approaches and results in the literature. Our hardness results can be stated in more details as follows.

\begin{theorem}\label{thm:P14free}
\MC, \PMC\ and \DPM\ are $\NP$-complete in $\{3P_6,2P_7,P_{14}\}$-free $8$-chordal graphs. 
Moreover, under the ETH, no algorithm with runtime $2^{o(n)}$ can solve any of these problems for $n$-vertex $\{3P_6,2P_7,P_{14}\}$-free $8$-chordal input graphs. 
\end{theorem}

Actually, we prove the following slightly stronger result: within $\{3P_6,2P_7,\allowbreak P_{14}\}$-free $8$-chordal graphs, it is hard to distinguish between those without matching cuts (respectively perfect matching cuts, disconnected perfect matchings) and those in which every matching cut is a perfect matching cut. Moreover, under the ETH, this task cannot be solved in subexponential time in the vertex number of the input graph.

On the positive side, we prove the following.
\begin{theorem}\label{thm:longholefree-dpm}
\DPM\ and \PMC\ can be solved in polynomial time when restricted to $4$-chordal graphs.
\end{theorem}
Our polynomial-time algorithm solving \DPM\ in $4$-chordal graphs is based on simple, known forcing rules. We remark that the approach also works for \MC, providing another proof for Theorem~\ref{thm:longholefree-mc}. In case of \PMC\ in $4$-chordal graphs, we will give a non-trivial polynomial-time reduction to \textsc{2sat}, based on the distance-level structure of $4$-chordal graphs admitting perfect matching cuts. 
Our positive and negative results in Theorems~\ref{thm:P14free} and \ref{thm:longholefree-dpm} on $k$-chordal graphs partly answer an open question stated in~\cite{LeT22,LuckePR23}. 

The current status of the computational complexity of \MC, \PMC\ and \DPM\ in 
$P_t$-free graphs and in $k$-chordal graphs are summarized in the Table~\ref{tab:status}. 

\begin{table}[h]
\begin{minipage}{.5\textwidth}
\begin{tabular}{@{}lclllcl@{}}
\toprule
$t$ & $\le 6$ & $7$ & $\cdots$ & $13$ & $\ge 14$\\
\midrule
\MC     & $\Po$  & ? & $\cdots$ & ? & $\NPC$\! {\small $(\star)$}  \\
\PMC   & $\Po$  & ? & $\cdots$ & ? & $\NPC$\! {\small $(\star)$}  \\
\DPM   & $\Po$  & ? & $\cdots$ & ? & $\NPC$\! {\small $(\star)$}  \\
\bottomrule
\end{tabular}
\end{minipage}
\qquad
\begin{minipage}{.3\textwidth}
\begin{tabular}{@{}lcllll@{}}
\toprule
$k$ & $\le 4$ & $5$ & $6$ &  $7$ & $\ge 8$\\
\midrule
\MC     & $\Po$   & ?  & ? & ? & $\NPC$\! {\small $(\star)$}  \\
\PMC   & $\Po$\! {\small $(\star)$}   & ?  & ? & ? & $\NPC$\! {\small $(\star)$}  \\
\DPM   & $\Po$\! {\small $(\star)$}   & ?  & ?  & ? & $\NPC$\! {\small $(\star)$}  \\
\bottomrule
\end{tabular}
\end{minipage}
\caption{The current status of the computational complexity of matching cut problems in $P_t$-free graphs (left hand-side) and in $k$-chordal graphs (right hand-side). $\Po$, $\NPC$, and `?' means polynomial, $\NP$-complete, and unknown, respectively. Results of this paper are marked with $(\star)$.}
\label{tab:status}
\end{table}

The paper is organized as follows. We recall some notions and notations in Section~\ref{sec:pre} which will be used. Then, we prove a slightly stronger result than Theorem~\ref{thm:P14free} in Section~\ref{sec:P14free}. The proof of Theorem~\ref{thm:longholefree-dpm} will be given in Section~\ref{sec:longholefree-dpm}. Section~\ref{sec:con} concludes the paper. 

\section{Preliminaries}\label{sec:pre}
The neighborhood of a vertex $v$ in a graph $G$ is denoted $N_G(v)$, or simply $N(v)$ if the context is clear. We write $N[v]$ for $N(v)\cup\{v\}$.  
For a set $\cal H$ of graphs, ${\cal H}$-free graphs are those in which no induced subgraph is isomorphic to a graph in $\cal H$. 
We denote by~$P_t$ the $t$-vertex path with $t-1$ edges and by $C_t$ the $t$-vertex cycle with $t$ edges. The $3$-vertex cycle~$C_3$ is also called a triangle. 
The union $G+H$ of two vertex-disjoint graphs~$G$ and~$H$ is the graph with vertex set $V(G)\cup V(H)$ and edge set $E(G)\cup E(H)$; we write $pG$ for the union of~$p$ copies of~$G$. 
For a subset $S \subseteq V(G)$, let $G[S]$ denote the subgraph of~$G$ induced by~$S$; $G-S$ stands for $G[V(G)\setminus S]$.  
By \lq $G$ contains an $H$\rq\ we mean~$G$ contains~$H$ as an induced subgraph. 

Given a matching cut $M=(X,Y)$ of a graph $G$, a vertex set $S\subseteq V(G)$ is \emph{monochromatic} if all vertices of~$S$ belong to the same \emph{part} of $M$, i.e., $S\subseteq X$ or else $S\subseteq Y$. Notice that every clique with at least $3$ vertices and the vertex set of any subgraph isomorphic to the complete bipartite graph $K_{2,s}$ with $s\ge 3$ are monochromatic with respect to any matching cut. Moreover, if $S$ is monochromatic and~$v$ is a vertex outside $S$ adjacent to at least two vertices in $S$ then $S\cup\{v\}$ is monochromatic.  

If $G$ has a perfect matching cut $M=(X,Y)$ then $N[v]$ is not monochromatic for any vertex $v$. 
Indeed, every vertex $v\in X$ has exactly one neighbor in $Y$ which we sometimes call the \emph{private neighbor} of $v$ and denote $p(v)$, and vice versa, every vertex $w\in Y$ has its private neighbor $p(w)$ in $X$. Thus, for any vertex $v$, $p(p(v))=v$.

Algorithmic lower bounds in this paper are conditional, based on the Exponential Time Hypothesis (ETH)~\cite{ImpagliazzoP01}. The ETH asserts that no algorithm can solve \threeSAT\ in subexponential time $2^{o(n)}$ for $n$-variable \textsc{3-cnf} formulas: there exists a constant $\delta>0$ such that no algorithm can solve \threeSAT\ in $2^{\delta n}$ time for \textsc{3-cnf} formulas over $n$ variables. 
As shown by the Sparsification Lemma in~\cite{ImpagliazzoPZ01}, the hard cases of \threeSAT\ already consist of sparse formulas with $m=O(n)$ clauses. Hence, the ETH implies that \threeSAT\ cannot be solved in $2^{o(n+m)}$ time.  

Recall that an instance for \oneIN3SAT\ is  a \textsc{3-cnf} formula $\phi=C_1\land C_2\land \cdots\allowbreak \land C_m$ over $n$ variables, in which each clause $C_j$ consists of three distinct literals. The problem asks whether there is a truth assignment to the variables such that every clause in $\phi$ has \emph{exactly one} true literal. We call such an assignment a \emph{1-in-3 assignment}. 
There is a polynomial reduction from \threeSAT\ to \oneIN3SAT\ (\cite[Theorem 7.2]{Moret98}, see also~\cite{Schaefer78}), which transforms an instance for \threeSAT\ to an equivalent instance for \oneIN3SAT\, linear in the number of variables and clauses. Thus, assuming ETH, \oneIN3SAT\ cannot be solved in $2^{o(n+m)}$ time on inputs with~$n$ variables and~$m$ clauses. 
We will need a restriction of \oneIN3SAT, \P1IN3SAT, in which each variable occurs only positively. 
There is a well-known polynomial reduction from \oneIN3SAT\ to \P1IN3SAT, which transforms an instance for \oneIN3SAT\ to an equivalent instance for \P1IN3SAT, linear in the number of variables and clauses. (See also an exercise in \cite[Exercise 7.1]{Moret98}.) Hence, we obtain: assuming ETH, 
\P1IN3SAT\ cannot be solved in $2^{o(n+m)}$ time for inputs with $n$ variables and $m$ clauses.

\section{Hardness results: Proof of Theorem~\ref{thm:P14free}}\label{sec:P14free} 
Recall that a perfect matching cut is in particular a matching cut, as well as a disconnected perfect matching. 
This observation leads to the following promise versions of \MC, \PMC\ and \DPM. (We refer to~\cite{Goldreich06a,Goldreich2010} for background on promise problems.)

\medskip\noindent
\fbox{
\begin{minipage}{.96\textwidth}
\textsc{promise-pmc mc}\\[.7ex]
\begin{tabular}{l l}
{\em Instance:\/}& A graph $G$ that either has no matching cut, or every\\
                         & matching cut is a perfect matching cut.\\
{\em Question:\/}& Does $G$ have a matching cut\,?\\
\end{tabular}
\end{minipage}
}

\medskip\noindent
\fbox{
\begin{minipage}{.96\textwidth}
\textsc{promise-pmc pmc}\\[.7ex]
\begin{tabular}{l l}
{\em Instance:\/}& A graph $G$ that either has no perfect matching cut, or every\\
                         & matching cut is a perfect matching cut.\\
{\em Question:\/}& Does $G$ have a perfect matching cut\,?\\
\end{tabular}
\end{minipage}
}

\medskip\noindent
\fbox{
\begin{minipage}{.96\textwidth}
\textsc{promise-pmc dpm}\\[.7ex]
\begin{tabular}{l l}
{\em Instance:}& A graph $G$ that either has no disconnected perfect matching, \\
                         & or every matching cut is a perfect matching cut.\\
{\em Question:}& Does $G$ have a disconnected perfect matching\,?\\
\end{tabular}
\end{minipage}
}

\medskip\noindent
In all the promise versions above, we are allowed not to consider certain input graphs. 
In \textsc{promise-pmc mc}, \textsc{promise-pmc pmc} and \textsc{promise-pmc dpm}, we are allowed to ignore graphs having a matching cut that is not a perfect matching cut, for which \MC\ must answer \lq yes\rq, and \PMC\ and \DPM\ may answer \lq yes\rq\ or \lq no\rq. 

We slightly improve Theorem~\ref{thm:P14free} by showing the following result.
\begin{theorem}\label{thm:promise}
\emph{\textsc{promise-pmc mc}}, \emph{\textsc{promise-pmc pmc}} and \emph{\textsc{promise-pmc dpm}} are $\NP$-complete, even when restricted to $\{3P_6,2P_7,P_{14}\}$-free graphs. 
Moreover, under the ETH, no algorithm with runtime $2^{o(n)}$ can solve any of these problems for $n$-vertex $\{3P_6,2P_7,P_{14}\}$-free $8$-chordal graphs. 
\end{theorem}

Clearly, Theorem~\ref{thm:promise} implies Theorem~\ref{thm:P14free}. 
Theorem~\ref{thm:promise} shows in particular that distinguishing between graphs without matching cuts and graphs in which every matching cut is a perfect matching cut is hard, and not only between those without matching cuts  and those with matching cuts which is implied by the $\NP$-completeness of \MC. Similar implications of Theorem~\ref{thm:promise} can be derived for \PMC\ and \DPM.

\subsection{The reduction}
We give a polynomial-time reduction from \P1IN3SAT\ to \textsc{promi\-se-pmc pmc} (and to \textsc{promise-pmc mc}, \textsc{promise-pmc dpm} at the same time). 

Let $\phi$ be a \textsc{3-cnf} formula with $m$ clauses $C_j$, $1\le j\le m$, and $n$ variables~$x_i$, $1\le i\le n$, in which each clause~$C_j$ consists of three distinct variables.  
We will construct a $\{3P_6,2P_7,P_{14}\}$-free $8$-chordal graph~$G$ such that~$G$ has a perfect matching cut if and only if $\phi$ admits a 1-in-3 assignment. Moreover, every matching cut of~$G$, if any, is a perfect matching cut.

A general idea in constructing a graph without long induced paths is to ensure that long paths must pass through few cliques, say at most two. The following observation gives hints on how to create such \lq path-blocking\rq\ cliques: consider a graph, in which the seven vertices $c, c_k, a_k$, $1\le k\le 3$, induce a tree with leaves $a_1,a_2,a_3$ and degree-2 vertices $c_1,c_2,c_3$ and the degree-3 vertex $c$. If the graph has a perfect matching cut, then $a_1, a_2,a_3$ must belong to the same part of the cut. Therefore, we can make $\{a_1,a_2,a_3\}$ adjacent to a clique and the resulting graph still has a perfect matching cut. 
Now, a gadget $G(H;v)$ may be obtained from a suitable graph~$H$ with $v\in V(H)$ as follows. Let~$H$ be a graph having a degree-3 vertex~$v$ with neighbors $b_1, b_2$ and~$b_3$. Let $G(H;v)$ be the graph obtained from $H-v$ by adding~7 new vertices $a_1,a_2,a_3,\allowbreak c_1,c_2,c_3$ and~$c$, and edges $cc_k$, $c_ka_k$, $a_kb_k$, $1\le k\le 3$, and $a_1a_2, a_1a_3$ and $a_2a_3$. (Thus, contracting the triangle $a_1a_2a_3$ from $G(H;v)-\{c,c_1,c_2,c_3\}$ we obtain the graph~$H$.) See Fig.~\ref{fig:G(H;v)}.  

\begin{figure}[ht]
\centering
\tikzstyle{vertex}=[draw,circle,inner sep=.8pt,fill=black] 
\begin{tikzpicture}[scale=.27] 

\filldraw[fill=hellgrau, opacity=0.35, draw opacity=0.15] (-1,0) rectangle (9,9);
\node at (4,7) {\small $H$}; 

\node[vertex] (v) at (4,1)  [label=right:{\small $v$}] {};
\node[vertex] (b1) at (1,4)   [label=left:{\small $b_{1}$}] {};
\node[vertex] (b2) at (4,4)  [label=right:{\small $b_{2}$}] {};  
\node[vertex] (b3) at (7,4)  [label=right:{\small $b_{3}$}] {};
\node (1) at (1,5)  {};
\node (2) at (2,5)  {}; 
\node (3) at (3,5)  {};
\node (4)  at (5,5) {}; 
\node (5) at (6,5)  {}; 
\node (6) at (7,5)  {}; 

\draw (v)--(b1); \draw (v)--(b2); \draw (v)--(b3); 
\draw[thin] (1)--(b1)--(2); \draw[thin] (3)--(b2)--(4); \draw[thin] (5)--(b3)--(6);
\end{tikzpicture} 
\quad
\begin{tikzpicture}[scale=.27] 

\filldraw[fill=hellgrau, opacity=0.35, draw opacity=0.15] (-1,9) rectangle (9,15);
\node at (4,13) {\small $H-v$}; 

\node[vertex] (a1) at (1,7)   [label=left:{\small $a_{1}$}] {};
\node[vertex] (a2) at (4,6) {}; 
\node at (5,5.5) {\small $a_{2}$};
\node[vertex] (a3) at (7,7)  [label=right:{\small $a_{3}$}] {};
\node[vertex] (b1) at (1,10)   [label=left:{\small $b_{1}$}] {};
\node[vertex] (b2) at (4,10)  [label=right:{\small $b_{2}$}] {};  
\node[vertex] (b3) at (7,10)  [label=right:{\small $b_{3}$}] {};
\node[vertex] (c1) at (1,4)   [label=left:{\small $c_{1}$}] {};
\node[vertex] (c2) at (4,4)  [label=right:{\small $c_{2}$}] {};  
\node[vertex] (c3) at (7,4)  [label=right:{\small $c_{3}$}] {};
\node[vertex] (c) at (4,1) [label=right:{\small $c$}] {};
\node (1) at (1,11)  {};
\node (2) at (2,11)  {}; 
\node (3) at (3,11)  {};
\node (4)  at (5,11) {}; 
\node (5) at (6,11)  {}; 
\node (6) at (7,11)  {}; 

\draw (c)--(c1); \draw (c)--(c2); \draw (c)--(c3); 
\draw (c1)--(a1)--(b1); \draw (c2)--(a2)--(b2); \draw (c3)--(a3)--(b3); 
\draw (a1)--(a2)--(a3)--(a1);
\draw[thin] (1)--(b1)--(2); \draw[thin] (3)--(b2)--(4); \draw[thin] (5)--(b3)--(6);
\end{tikzpicture} 
\quad
\begin{tikzpicture}[scale=.27] 
\node[vertex] (c1) at (4,1)  [label=right:{\small $c$}] {};
\node[vertex] (c11) at (1,4)   [label=left:{\small $c_{1}$}] {};
\node[vertex] (c12) at (4,4)  [label=right:{\small $c_{2}$}] {};  
\node[vertex] (c13) at (7,4)  [label=right:{\small $c_{3}$}] {};
\node[vertex] (a11) at (1,7)  [label=left:{\small $a_{1}$}] {};
\node[vertex] (a12) at (4,6)  {}; 
\node at (5,5.5) {\small $a_{2}$};
\node[vertex] (a13) at (7,7)  [label=right:{\small $a_{3}$}] {};
\node[vertex] (b11)  at (1,10) {}; 
\node[vertex] (b12) at (4,10)  {};
\node[vertex] (b13) at (7,10)  {}; 
\node[vertex] (1)  at (2,12) {}; 
\node[vertex] (2) at (6,12)  {}; 
\node[vertex] (3) at (2.5,15)  {}; 
\node[vertex] (4)  at (1,17) {}; 
\node[vertex] (5)  at (7,17) {}; 
\node[vertex] (6)  at (5.5,15) {}; 

\draw (c1)--(c11); \draw (c1)--(c12); \draw (c1)--(c13); 
\draw (c11)--(a11)--(b11)--(1)--(2)--(3)--(4)--(5)--(6)--(1);
\draw (c12)--(a12)--(b12)--(3); \draw (b12)--(6); \draw (b11)--(4);
\draw (c13)--(a13)--(b13)--(2); \draw (b13)--(5);
\draw (a11)--(a12)--(a13)--(a11);
\end{tikzpicture} 
\quad
\begin{tikzpicture}[scale=.27] 
\node[vertex] (c1) at (4,1)  [label=right:{\small $c$}] {};
\node[vertex] (c11) at (1,4)   [label=left:{\small $c_{1}$}] {};
\node[vertex] (c12) at (4,4)  [label=right:{\small $c_{2}$}] {};  
\node[vertex] (c13) at (7,4)  [label=right:{\small $c_{3}$}] {};
\node[vertex] (a11) at (1,7)  [label=left:{\small $a_{1}$}] {};
\node[vertex] (a12) at (4,6)  {}; 
\node at (5,5.5) {\small $a_{2}$};
\node[vertex] (a13) at (7,7)  [label=right:{\small $a_{3}$}] {};
\node[vertex] (b11)  at (1,10) {}; 
\node[vertex] (b12) at (4,10)  {}; 
\node[vertex] (b13) at (7,10)  {}; 
\node[vertex] (1)  at (2,12) {}; 
\node[vertex] (2) at (3.3,12)  {}; 
\node[vertex] (3) at (4.7,12)  {}; 
\node[vertex] (4)  at (6,12) {}; 
\node[vertex] (5)  at (7,15) {}; 
\node[vertex] (6)  at (1,15) {}; 

\draw (c1)--(c11); \draw (c1)--(c12); \draw (c1)--(c13); 
\draw (c11)--(a11)--(b11)--(1)--(2)--(b12)--(3)--(4)--(b13)--(5)--(6)--(b11);
\draw (a12)--(b12); \draw (a13)--(b13); \draw (c12)--(a12); \draw (c13)--(a13);
\draw (a11)--(a12)--(a13)--(a11);
\end{tikzpicture} 
\caption{From left to right: the graph $H$ with a degree-3 vertex $v$, the graph $G(H;v)$, the graph $G(H;v)$ where $H$ is the Petersen graph, and the graph $G(H;v)$ where $H$ is the Heggernes-Telle graph.
}\label{fig:G(H;v)}
\end{figure}
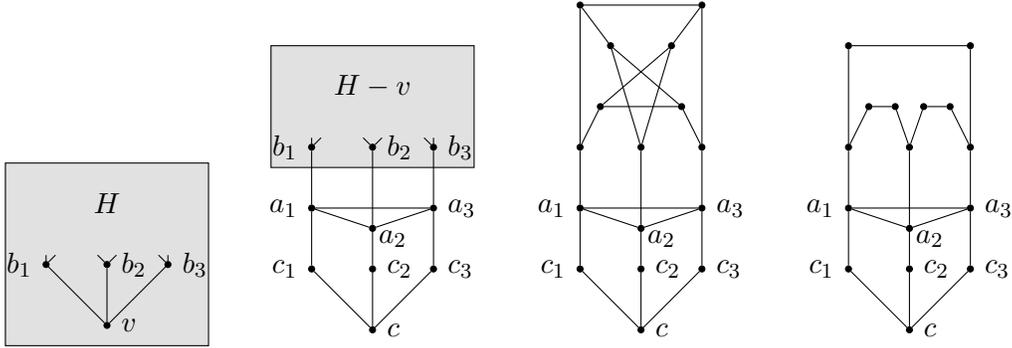

\begin{observation}\label{obs:H}
Assume that, for \emph{any} neighbor $b_i$ of $v$ in $H$, $H$ has a perfect matching cut $(X,Y)$ such that $v\in X$ and $b_i\in Y$. Then, for \emph{any} neighbor $d$ of $c$ in $G(H;v)$, the graph $G(H;v)$ has a perfect matching cut $(X',Y')$ such that $c\in X'$ and $d\in Y'$.
\end{observation}
Examples of graphs $H$ in Observation~\ref{obs:H} include the cube, the Petersen graph and the 10-vertex Heggernes-Telle graph in~\cite[Fig. 3.1 with $v=a_i^1$]{HeggernesT98}. 
Our gadget will be obtained by taking the cube.\footnote{Taking the Petersen graph or the Heggernes-Telle graph will produce induced $P_{t}$ for some $t\ge 15$. If there exists another graph~$H$ \lq better\rq\ than the cube, then our construction will yield a $P_t$-free graph for some~$10\le t\le 13$.} 

We now formally describe the reduction. For each clause $C_j$ consisting of three variables $c_{j1}, c_{j2}$ and $c_{j3}$, let $G(C_j)$ be the graph depicted in Fig.~\ref{fig:P14freeGadget}, the graph $G(H;v)$ with the cube~$H$. We call $c_j$ and~$c_j'$ the \emph{clause vertices}, and $c_{j1}$, $c_{j2}$ and~$c_{j3}$ the \emph{variable vertices}. 
Then, the graph~$G$ is obtained from all $G(C_j)$ by adding 
\begin{itemize}
\item all possible edges between variable vertices $c_{jk}$ and $c_{j'k'}$ of the same variable. Thus, for each variable~$x$, 
\begin{align*}
Q(x)=\,&\{c_{jk}\mid 1\le j\le m, 1\le k\le 3, \text{$x$ occurs in clause $C_j$ as $c_{jk}$}\}
\end{align*} 
is a clique in $G$, 
\end{itemize}
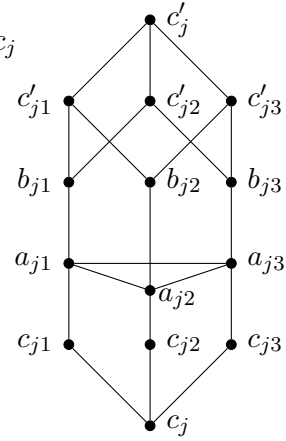
\begin{wrapfigure}{r}{.27\textwidth}
\vspace{-7.8cm} 
\centering
\,\,\begin{tikzpicture}[scale=.36] 
\tikzstyle{vertexS}=[draw,circle,inner sep=2pt,fill=black]
\tikzstyle{vertex}=[draw,circle,inner sep=1.3pt,fill=black] 
\node[vertex] (c1) at (4,1)  [label=right:{\small $c_j$}] {};
\node[vertex] (c11) at (1,4)   [label=left:{\small $c_{j1}$}] {};
\node[vertex] (c12) at (4,4)  [label=right:{\small $c_{j2}$}] {};  
\node[vertex] (c13) at (7,4)  [label=right:{\small $c_{j3}$}] {};
\node[vertex] (a11) at (1,7)  [label=left:{\small $a_{j1}$}] {};
\node[vertex] (a12) at (4,6)  {}; 
\node at (5.25,5.6) {\small $a_{j2}$};
\node[vertex] (a13) at (7,7)  [label=right:{\small $a_{j3}$}] {};
\node[vertex] (b11)  at (1,10) [label=left:{\small $b_{j1}$}] {};
\node[vertex] (b12) at (4,10)  [label=right:{\small $b_{j2}$}] {};
\node[vertex] (b13) at (7,10)  [label=right:{\small $b_{j3}$}] {};
\node[vertex] (c11')  at (1,13) [label=left:{\small $c_{j1}'$}] {};
\node[vertex] (c12') at (4,13)  [label=right:{\small $c_{j2}'$}] {};
\node[vertex] (c13') at (7,13)  [label=right:{\small $c_{j3}'$}] {};
\node[vertex] (c1')  at (4,16) [label=right:{\small $c_{j}'$}] {};

\draw (c1)--(c11); \draw (c1)--(c12); \draw (c1)--(c13); 
\draw (c11)--(a11)--(b11)--(c11')--(c1');
\draw (c12)--(a12)--(b12)--(c11'); \draw (b12)--(c13'); \draw (c1')--(c12');
\draw (c13)--(a13)--(b13)--(c13')--(c1');
\draw (b11)--(c12')--(b13);
\draw (a11)--(a12)--(a13)--(a11);
\end{tikzpicture} 
\caption{The gadget $G(C_j)$. 
}\label{fig:P14freeGadget}
\end{wrapfigure}
\begin{itemize}

\item all possible edges between the~$2m$ clause vertices~$c_{j}$\\ and~$c_{j}'$. 
Thus, 
\begin{align*}
F=\, &\{c_j\mid 1\le j\le m\}\cup\{c_j'\mid 1\le j\le m\}
\end{align*} 
is a clique in $G$, 

\item all possible edges between the~$3m$ vertices~$a_{jk}$.\\ 
Thus, 
\begin{align*}
T=\{a_{jk}\mid 1\le j\le m, 1\le k\le 3\}
\end{align*} 
is a clique  in $G$.
\end{itemize}

\medskip
\noindent
The description of $G$ is complete. As an example, the graph $G$ from the formula~$\phi$ with three clauses $C_1=\{x,y,z\}, C_2=\{u,z,y\}$ and $C_3=\{z,v,w\}$ is depicted in Fig.~\ref{fig:P14freeExample1}.

\begin{figure}[ht]
\begin{center}
\tikzstyle{vertexS}=[draw,circle,inner sep=2pt,fill=black]
\tikzstyle{vertex}=[draw,circle,inner sep=1.3pt,fill=black] 
\tikzstyle{vertexR}=[draw=black,circle,inner sep=1.3pt,fill=flax] 
\tikzstyle{vertexB}=[draw=black,circle,inner sep=1.3pt,fill=teal] 

\begin{tikzpicture}[scale=.36] 
\filldraw[fill=teal, opacity=0.10, draw opacity=0.15] (-0.5,5.1) rectangle (27,8);
\node at (0,5.6) {\small \textcolor{gray}{$T$}}; 
\filldraw[fill=flax, opacity=0.20, draw opacity=0.15] (1.5,15.3) rectangle (24.5,16.8);
\node at (2,15.8) {\small \textcolor{gray}{$F$}}; 
\filldraw[fill=flax, opacity=0.20, draw opacity=0.15] (1.5,0.3) rectangle (24.5,1.8);
\node at (2,0.8) {\small \textcolor{gray}{$F$}}; 

\node[vertexR] (c1) at (4,1)  [label=right:{\small $c_{1}$}]{};
\node[vertex] (c11) at (1,4)  [label=right:{\small $x$}] {};
\node[vertex] (c12) at (4,4)  [label=left:{\small $y$}]{};  
\node[vertex] (c13) at (7,4)  [label=left:{\small $z$}] {};
\node[vertexB] (a11) at (1,7)  {};
\node at (1.95,7.5) {\small $a_{11}$};
\node[vertexB] (a12) at (4,6)  {};
\node at (4.95,5.5) {\small $a_{12}$};
\node[vertexB] (a13) at (7,7)  {}; 
\node at (7.95,7.5) {\small $a_{13}$};
\node[vertex] (b11)  at (1,10) [label=right:{\small $b_{11}$}] {};
\node[vertex] (b12) at (4,10)  [label=right:{\small $b_{12}$}] {};
\node[vertex] (b13) at (7,10)  [label=right:{\small $b_{13}$}] {};
\node[vertex] (c11')  at (1,13) [label=right:{\small $c_{11}'$}] {};
\node[vertex] (c12') at (4,13)  [label=right:{\small $c_{12}'$}] {};
\node[vertex] (c13') at (7,13)  [label=right:{\small $c_{13}'$}] {};
\node[vertexR] (c1')  at (4,16) [label=right:{\small $c_{1}'$}]{};

\draw (c1)--(c11); \draw (c1)--(c12); \draw (c1)--(c13); 
\draw (c11)--(a11)--(b11)--(c11')--(c1');
\draw (c12)--(a12)--(b12)--(c11'); \draw (b12)--(c13'); \draw (c1')--(c12');
\draw (c13)--(a13)--(b13)--(c13')--(c1');
\draw (b11)--(c12')--(b13);
\draw (a11)--(a12)--(a13)--(a11);

\node[vertexR] (c2) at (13,1)  [label=right:{\small $c_{2}$}] {};
\node[vertex] (c21) at (10,4)   [label=right:{\small $u$}] {};
\node[vertex] (c22) at (13,4)  [label=right:{\small $z$}] {}; 
\node[vertex] (c23) at (16,4)  [label=right:{\small $y$}] {};
\node[vertexB] (a21) at (10,7)  {};
\node at (10.95,7.5) {\small $a_{21}$};
\node[vertexB] (a22) at (13,6)  {};
\node at (13.95,5.5) {\small $a_{22}$};
\node[vertexB] (a23) at (16,7)  {};
\node at (16.95,7.5) {\small $a_{23}$};
\node[vertex] (b21)  at (10,10) [label=right:{\small $b_{21}$}]{};
\node[vertex] (b22) at (13,10)  [label=right:{\small $b_{22}$}]{};
\node[vertex] (b23) at (16,10)  [label=right:{\small $b_{23}$}]{};
\node[vertex] (c21')  at (10,13)  [label=right:{\small $c_{21}'$}] {}; 
\node[vertex] (c22') at (13,13)   [label=right:{\small $c_{22}'$}] {};
\node[vertex] (c23') at (16,13)   [label=right:{\small $c_{23}'$}] {};
\node[vertexR] (c2')  at (13,16) [label=right:{\small $c_{2}'$}]{};

\draw (c2)--(c21); \draw (c2)--(c22); \draw (c2)--(c23); 
\draw (c21)--(a21)--(b21)--(c21')--(c2');
\draw (c22)--(a22)--(b22)--(c21'); \draw (b22)--(c23'); \draw (c2')--(c22');
\draw (c23)--(a23)--(b23)--(c23')--(c2');
\draw (b21)--(c22')--(b23);
\draw (a21)--(a22)--(a23)--(a21);

\node[vertexR] (c3) at (22,1)   [label=right:{\small $c_{3}$}] {};
\node[vertex] (c31) at (19,4)   [label=right:{\small $z$}] {};
\node[vertex] (c32) at (22,4) [label=right:{\small $v$}] {};  
\node[vertex] (c33) at (25,4)    [label=right:{\small $w$}] {};
\node[vertexB] (a31) at (19,7)  {};
\node at (19.95,7.5) {\small $a_{31}$};
\node[vertexB] (a32) at (22,6)  {};
\node at (22.95,5.5) {\small $a_{32}$};
\node[vertexB] (a33) at (25,7)  {};
\node at (25.95,7.5) {\small $a_{33}$};
\node[vertex] (b31)  at (19,10) [label=right:{\small $b_{31}$}]{};
\node[vertex] (b32) at (22,10)  [label=right:{\small $b_{32}$}]{};
\node[vertex] (b33) at (25,10)  [label=right:{\small $b_{33}$}]{};
\node[vertex] (c31')  at (19,13) [label=right:{\small $c_{31}'$}]{};
\node[vertex] (c32') at (22,13)  [label=right:{\small $c_{32}'$}]{};
\node[vertex] (c33') at (25,13)  [label=right:{\small $c_{33}'$}]{};
\node[vertexR] (c3')  at (22,16) [label=right:{\small $c_{3}'$}]{};

\draw (c3)--(c31); \draw (c3)--(c32); \draw (c3)--(c33); 
\draw (c31)--(a31)--(b31)--(c31')--(c3');
\draw (c32)--(a32)--(b32)--(c31'); \draw (b32)--(c33'); \draw (c3')--(c32');
\draw (c33)--(a33)--(b33)--(c33')--(c3');
\draw (b31)--(c32')--(b33);
\draw (a31)--(a32)--(a33)--(a31);

\draw (c12)[bend angle=30, bend right]to (c23); 
\draw (c13)[bend angle=28, bend left]to (c22); \draw (c22)[bend angle=30, bend left]to (c31);
\draw (c13)[bend angle=25, bend right]to (c31);

\end{tikzpicture} 
\caption{The graph $G$ from the formula $\phi$ with three clauses $C_1=\{x,y,z\}$, $C_2=\{u,z,y\}$ and $C_3=\{z,v,w\}$.  
The 6 flax vertices $c_1, c_2, c_3, c_1', c_2', c_3'$ and the 9 teal vertices $a_{11}$, $a_{12}$, $a_{13}$, $a_{21}$, $a_{22}$, $a_{23}$, $a_{31}$, $a_{32}$, $a_{33}$ form the clique $F$ and $T$, respectively.}\label{fig:P14freeExample1}
\end{center}
\end{figure}
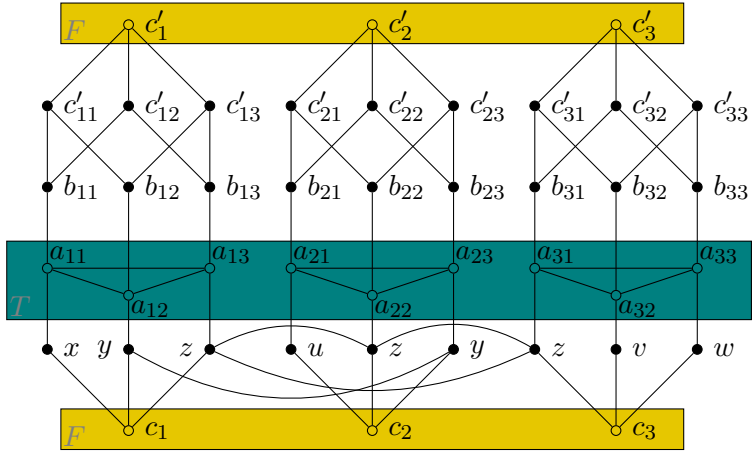

Notice that no edge exists between the two \lq path-blocking\rq\ cliques $F$ and~$T$. 
Notice also that $G-F-T$ has exactly $m+n$ components: 

\begin{itemize}
\item For each $1\le j\le m$, the $6$-cycle $D_j:\, b_{j1}, c_{j1}', b_{j2}, c_{j3}', b_{j3}, c_{j2}'$
is a component of $G-F-T$, call it the clause $6$-cycle (of clause $C_j$);  
\item For each variable $x$, the clique $Q(x)$
is a component of $G-F-T$, call it the variable clique (of variable $x$).
\end{itemize}

\begin{lemma}\label{lem:P14free}
$G$ is $\{3P_6,2P_7,P_{14}\}$-free and $8$-chordal.
\end{lemma}
\begin{proof}
First, observe that each component of $G-F-T$ is a clause $6$-cycle $D_j$ or a variable clique $Q(x)$. Hence,  
\begin{align}\label{eq:1}
&\text{$G-F-T$ is $P_6$-free.}
\end{align}
Therefore, every induced $P_6$ in $G$ must contain a vertex from the clique~$F$ or from the clique~$T$. This shows that $G$ is $3P_6$-free. 

Observe next that, for each $j$, $c_j'\in F$ is the cut-vertex in $G-T$ separating the clause $6$-cycle $D_j$ and $F$, and $N(c_j')\cap D_j =\{c_{j1}', c_{j2}', c_{j3}'\}$. 
Observe also that, for each~$x$, $(G-T)[Q(x)\cup F]$ is a co-bipartite graph, the complement of a bipartite graph. Hence, it can be verified immediately that
\begin{align}\label{eq:2}
&\text{$G-T$ is $P_7$-free.} 
\end{align}
Fact~(\ref{eq:2}) implies that every induced $P_7$ in $G$ must contain a vertex from the clique $T$. This shows that $G$ is $2P_7$-free.

We now are ready to argue that $G$ is $P_{14}$-free. Suppose not and let $P: v_1,v_2,\allowbreak \ldots, v_{14}$ be an induced $P_{14}$ in $G$, with edges $v_iv_{i+1}$, $1\le i<14$. For $i< j$, write $P[v_i,v_j]$ for the subpath of~$P$ between (and including) $v_i$ and~$v_j$. 
Then, by~(\ref{eq:2}), each of $P[v_1,v_7]$ and $P[v_8,v_{14}]$ contains a vertex from the clique $T$. Since~$P$ has no chords, $P[v_1,v_7]$ has only the vertex~$v_7$ in $T$ and $P[v_8,v_{14}]$ has only the vertex~$v_8$ in $T$. By~(\ref{eq:1}), therefore, both $P[v_1,v_6]$ and $P[v_9,v_{14}]$ contain some vertex in the clique $F$, and thus~$P$ has a chord. 
This contradiction shows that~$G$ is $P_{14}$-free, as claimed.

We now show that $G$ is $8$-chordal. First observe that, if $S$ is a clique cutset in (an arbitrary graph) $G$ then any induced cycle in $G$ is contained in $G[A\cup S]$ for some connected component $A$ of $G-S$.  With this observation, we have
\begin{align}\label{eq:3}
& \text{$G-T$ has no induced cycles $C_k$ of length $k\ge 7$.}
\end{align}
This can be seen as follows:  
write $F=C\cup C'$ with $C=\{c_j\mid 1\le j\le m\}$ and $C'=\{c_j'\mid 1\le j\le m\}$. Then $C$ is a clique cutset of $G-T$ separating each variable clique $Q(x)$ and the connected component $H$ consisting of $C'$ and all clause $6$-cycles~$D_j$. Hence each induced cycle in $G-T$ is contained in some $Q(x)\cup C$ or in $H\cup C$. Since $Q(x)\cup C$ induces a co-bipartite graph, it has no induced cycle of length larger than~$4$. Furthermore, $C'$ is a clique cutset of the graph induced by $H\cup C$ separating each~$D_j$ and~$C$. Hence each induced cycle in $H\cup C$ is contained in some $D_j\cup C'$ or in $C\cup C'=F$. In each $D_j\cup C'$, $c_j'$ is the cut-vertex separating $D_j$ and $C'\setminus\{c_j'\}$. Since~$c_j'$ is adjacent to three vertices of the $6$-cycle $D_j$ and $F$ is a clique, $D_j\cup C'$ and hence $H\cup C'$ has no induced cycles of length larger than~$6$. 
Thus, (\ref{eq:3}) is proved.

Now, suppose for the contrary that $G$ has an induced cycle $Z$ of length at least~ $9$. Then, by (\ref{eq:3}), $Z\cap T\not=\emptyset$. Since $T$ is a clique, $Z\cap T$ is a vertex or two adjacent vertices. But then $Z-T$ is an induced path in $G-T$ with at leat~$7$ vertices, contradicting (\ref{eq:2}). Thus, $G$ is $8$-chordal.
\qed
\end{proof}

We remark that there are many induced $P_{13}$ and $C_8$ in $G$; for instance, some $P_{13}$ in the graph depicted in Fig.~\ref{fig:P14freeExample1} are: 
\begin{itemize}
\item $b_{11}, c_{11}', b_{12}, c_{13}', b_{13}, a_{13}, a_{22}, c_{22}=z, c_{31}=z, c_3, c_1, c_{12}=y, c_{23}=y$;
\item $b_{11}, c_{11}', b_{12}, c_{13}', b_{13}, a_{13}, a_{21}, b_{21}, c_{21}', c_2', c_3', c_{33}', b_{33}$;
\item $b_{11}, c_{11}', b_{12}, c_{13}', b_{13}, a_{13}, a_{21}, b_{21}, c_{21}', c_2', c_2, c_{23}=y, c_{12}=y$, 
\end{itemize}

and some $C_8$ are:
\begin{itemize}
\item $a_{13}, a_{21}, b_{21}, c_{21}', c_{2}', c_{1}', c_{13}', b_{13}$; 
\item $a_{22}, a_{32}, b_{32}, c_{33}', c_{3}', c_{3}, c_{13}=z, c_{22}=z$, and 
\item $a_{12}, c_{12}=y, c_{23}=y, c_, c_3, c_{31}=z, c_{13}=z, a_{13}$.  
\end{itemize}

\begin{lemma}\label{lem:mc-in-G}
For any matching cut $M=(X,Y)$ of $G$, 
\begin{itemize}
\item[\em (i)] $F$ and $T$ are contained in different parts of $M$;
\item[\em (ii)] if $F\subseteq X$, then $|\{c_{j1}, c_{j2}, c_{j3}\}\cap Y|=1$, and if $F\subseteq Y$, then $|\{c_{j1}, c_{j2},\allowbreak c_{j3}\}\cap X|=1$;
\item[\em (iii)] for any variable $x$, $Q(x)$ is monochromatic; 
\item[\em (iv)]  if $F\subseteq X$, then $|\{b_{j1}, b_{j2}, b_{j3}\}\cap Y|=2$ and $|\{c_{j1}', c_{j2}', c_{j3}'\}\cap Y|=1$, and if $F\subseteq Y$, then $|\{b_{j1}, b_{j2}, b_{j3}\}\cap X|=2$ and $|\{c_{j1}', c_{j2}', c_{j3}'\}\cap X|=1$. 
\end{itemize}
\end{lemma}
\begin{proof}
Notice that $F$ and $T$ are cliques with at least three vertices, hence~$F$ and~$T$ are monochromatic. 

(i): Suppose not, and let $F$ and $T$ both be contained in $X$, say.  Then all variable vertices~$c_{jk}$, $1\le j\le m, 1\le k\le 3$, also belong to $X$ because each of them has two neighbors in $F\cup T\subseteq X$. 
Now, if all $b_{jk}$ are in~$X$, then all $c_{jk}'$ are also in~$X$ because in this case each of them has two neighbors in~$X$, and thus $X=V(G)$.      
Thus some $b_{jk}$ is in~$Y$, and so are its two neighbors in $\{c_{j1}',c_{j2}',c_{j3}'\}$. But then $c_j'$, which is in~$X$, has two neighbors in~$Y$. This contradiction shows that~$F$ and~$T$ must belong to different parts of~$M$, hence (i). 

(ii): By (i), let $F\subseteq X$ and $T\subseteq Y$, say. (The case $F\subseteq Y$ is symmetric.) Then, for any $j$, at most one of $c_{j1}, c_{j2}$ and~$c_{j3}$ can be outside~$X$. Assume that, for some~$j$, all $c_{j1}, c_{j2}, c_{j3}$ are in $X$. The assumption implies that all $b_{j1}, b_{j2}, b_{j3}$ belong to~$Y$, and then all $c_{j1}', c_{j2}', c_{j3}'$ belong to~$Y$, too. But then $c_j'$, which is in $X$, has three neighbors in $Y$. This contradiction shows (ii). 

(iii): Suppose that two variable vertices $c_{jk}$ and $c_{j'k'}$ in some clique $Q(x)$ are in different parts of $M$. Then, as $c_{jk}$ and $c_{j'k'}$ have neighbor $c_j$ and $c_{j'}$, respectively, in the monochromatic clique $F$, $c_{jk}$ has two neighbors in the part of~$c_{j'k'}$ or~$c_{j'k'}$ has two neighbors in the part of~$c_{jk}$. This contradiction shows (iii).

(iv): This fact can be derived from (i) and (ii). 
\qed 
\end{proof}

\begin{lemma}\label{lem:mc->pmc}
Every matching cut of $G$, if any, is a perfect matching cut.
\end{lemma}
\begin{proof}
Let $M=(X,Y)$ be a matching cut of $G$. By Lemma~\ref{lem:mc-in-G}~(i), let $F\subseteq X$ and $T\subseteq Y$, say. We argue that every vertex in~$X$ has a neighbor (hence exactly one) in~$Y$. Indeed, for each~$j$, 
\begin{itemize}
\item $c_j\in F\subseteq X$ has a neighbor $c_{jk}\in Y$ (by Lemma~\ref{lem:mc-in-G}~(ii)),
\item $c_j'\in F\subseteq X$ has a neighbor $c_{jk}'\in Y$ (by Lemma~\ref{lem:mc-in-G}~(iv)),
\item each $c_{jk}\in X$ has a neighbor $a_{jk}\in T\subseteq Y$ (by construction of~$G$),
\item each $b_{jk}\in X$ has a neighbor $a_{jk}\in T\subseteq Y$ (by construction of~$G$),
\item each $c_{jk}'\in X$ has a neighbor in $\{b_{j1},b_{j2},b_{j3}\}\cap Y$ (by Lemma~\ref{lem:mc-in-G}~(iv)).
\end{itemize}
Similarly, it can be seen that every vertex in $Y$ has a neighbor in $X$.\qed
\end{proof}

\begin{lemma}\label{lem:1in3->pmc}
If $\phi$ has a 1-in-3 assignment, then $G$ has a perfect matching cut.
\end{lemma}
\begin{proof}
Partition $V(G)$ into disjoint subsets $X$ and $Y$ as follows. (Fig.~\ref{fig:1in3->pmc} shows the partition for the example graph in Fig.~\ref{fig:P14freeExample1} given the assignment $\allowbreak y=v= \text{True}$, $x=z = u=w= \text{False}$.)  
First, 
\begin{itemize}
\item put $F$ into $X$, and for all variables $x$ which are assigned with False, put $Q(x)$ into $X$;
\item for each $1\le j\le m$, let $c_{jk}$ with $k=k(j)\in\{1,2,3\}$ be the variable vertex, for which the variable $x$ of $c_{jk}$ is assigned with True. Then put $b_{jk}$ and its two neighbors in $\{c_{j1}', c_{j2}', c_{j3}'\}$ into $X$.  
\end{itemize}
Let $Y=V(G)\setminus X$. Then, it is not difficult to verify that $M=(X,Y)$ is a perfect matching cut of $G$.\qed
\end{proof}

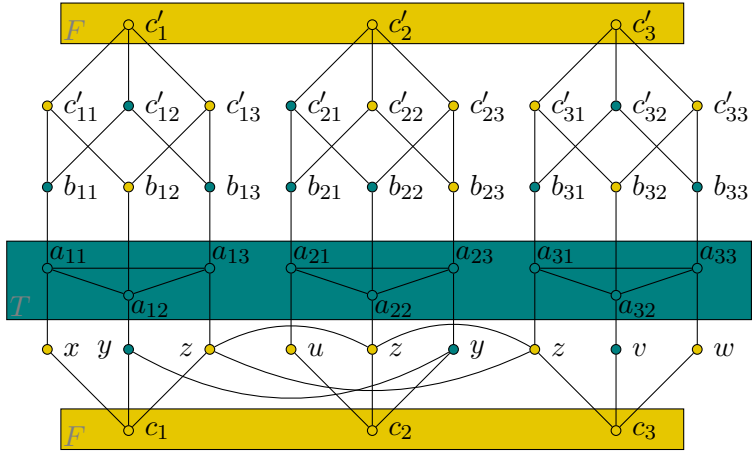
\begin{figure}[ht]
\begin{center}
\tikzstyle{vertexS}=[draw,circle,inner sep=2pt,fill=black]
\tikzstyle{vertex}=[draw,circle,inner sep=1.3pt,fill=black] 
\tikzstyle{vertexR}=[draw=black,circle,inner sep=1.3pt,fill=flax] 
\tikzstyle{vertexB}=[draw=black,circle,inner sep=1.3pt,fill=teal] 

\begin{tikzpicture}[scale=.36] 
\filldraw[fill=teal, opacity=0.10, draw opacity=0.15] (-0.5,5.1) rectangle (27,8);
\node at (0,5.6) {\small \textcolor{gray}{$T$}}; 
\filldraw[fill=flax, opacity=0.20, draw opacity=0.15] (1.5,15.3) rectangle (24.5,16.8);
\node at (2,15.8) {\small \textcolor{gray}{$F$}}; 
\filldraw[fill=flax, opacity=0.20, draw opacity=0.15] (1.5,0.3) rectangle (24.5,1.8);
\node at (2,0.8) {\small \textcolor{gray}{$F$}}; 

\node[vertexR] (c1) at (4,1)  [label=right:{\small $c_{1}$}]{};
\node[vertexR] (c11) at (1,4)  [label=right:{\small $x$}] {};
\node[vertexB] (c12) at (4,4)  [label=left:{\small $y$}]{};  
\node[vertexR] (c13) at (7,4)  [label=left:{\small $z$}] {};
\node[vertexB] (a11) at (1,7)  {};
\node at (1.95,7.5) {\small $a_{11}$};
\node[vertexB] (a12) at (4,6)  {};
\node at (4.95,5.5) {\small $a_{12}$};
\node[vertexB] (a13) at (7,7)  {}; 
\node at (7.95,7.5) {\small $a_{13}$};
\node[vertexB] (b11)  at (1,10) [label=right:{\small $b_{11}$}] {};
\node[vertexR] (b12) at (4,10)  [label=right:{\small $b_{12}$}] {};
\node[vertexB] (b13) at (7,10)  [label=right:{\small $b_{13}$}] {};
\node[vertexR] (c11')  at (1,13) [label=right:{\small $c_{11}'$}] {};
\node[vertexB] (c12') at (4,13)  [label=right:{\small $c_{12}'$}] {};
\node[vertexR] (c13') at (7,13)  [label=right:{\small $c_{13}'$}] {};
\node[vertexR] (c1')  at (4,16) [label=right:{\small $c_{1}'$}]{};

\draw (c1)--(c11); \draw (c1)--(c12); \draw (c1)--(c13); 
\draw (c11)--(a11)--(b11)--(c11')--(c1');
\draw (c12)--(a12)--(b12)--(c11'); \draw (b12)--(c13'); \draw (c1')--(c12');
\draw (c13)--(a13)--(b13)--(c13')--(c1');
\draw (b11)--(c12')--(b13);
\draw (a11)--(a12)--(a13)--(a11);

\node[vertexR] (c2) at (13,1)  [label=right:{\small $c_{2}$}] {};
\node[vertexR] (c21) at (10,4)   [label=right:{\small $u$}] {};
\node[vertexR] (c22) at (13,4)  [label=right:{\small $z$}] {}; 
\node[vertexB] (c23) at (16,4)  [label=right:{\small $y$}] {};
\node[vertexB] (a21) at (10,7)  {};
\node at (10.95,7.5) {\small $a_{21}$};
\node[vertexB] (a22) at (13,6)  {};
\node at (13.95,5.5) {\small $a_{22}$};
\node[vertexB] (a23) at (16,7)  {};
\node at (16.95,7.5) {\small $a_{23}$};
\node[vertexB] (b21)  at (10,10) [label=right:{\small $b_{21}$}]{};
\node[vertexB] (b22) at (13,10)  [label=right:{\small $b_{22}$}]{};
\node[vertexR] (b23) at (16,10)  [label=right:{\small $b_{23}$}]{};
\node[vertexB] (c21')  at (10,13)  [label=right:{\small $c_{21}'$}] {}; 
\node[vertexR] (c22') at (13,13)   [label=right:{\small $c_{22}'$}] {};
\node[vertexR] (c23') at (16,13)   [label=right:{\small $c_{23}'$}] {};
\node[vertexR] (c2')  at (13,16) [label=right:{\small $c_{2}'$}]{};

\draw (c2)--(c21); \draw (c2)--(c22); \draw (c2)--(c23); 
\draw (c21)--(a21)--(b21)--(c21')--(c2');
\draw (c22)--(a22)--(b22)--(c21'); \draw (b22)--(c23'); \draw (c2')--(c22');
\draw (c23)--(a23)--(b23)--(c23')--(c2');
\draw (b21)--(c22')--(b23);
\draw (a21)--(a22)--(a23)--(a21);

\node[vertexR] (c3) at (22,1)   [label=right:{\small $c_{3}$}] {};
\node[vertexR] (c31) at (19,4)   [label=right:{\small $z$}] {};
\node[vertexB] (c32) at (22,4) [label=right:{\small $v$}] {};  
\node[vertexR] (c33) at (25,4)    [label=right:{\small $w$}] {};
\node[vertexB] (a31) at (19,7)  {};
\node at (19.95,7.5) {\small $a_{31}$};
\node[vertexB] (a32) at (22,6)  {};
\node at (22.95,5.5) {\small $a_{32}$};
\node[vertexB] (a33) at (25,7)  {};
\node at (25.95,7.5) {\small $a_{33}$};
\node[vertexB] (b31)  at (19,10) [label=right:{\small $b_{31}$}]{};
\node[vertexR] (b32) at (22,10)  [label=right:{\small $b_{32}$}]{};
\node[vertexB] (b33) at (25,10)  [label=right:{\small $b_{33}$}]{};
\node[vertexR] (c31')  at (19,13) [label=right:{\small $c_{31}'$}]{};
\node[vertexB] (c32') at (22,13)  [label=right:{\small $c_{32}'$}]{};
\node[vertexR] (c33') at (25,13)  [label=right:{\small $c_{33}'$}]{};
\node[vertexR] (c3')  at (22,16) [label=right:{\small $c_{3}'$}]{};

\draw (c3)--(c31); \draw (c3)--(c32); \draw (c3)--(c33); 
\draw (c31)--(a31)--(b31)--(c31')--(c3');
\draw (c32)--(a32)--(b32)--(c31'); \draw (b32)--(c33'); \draw (c3')--(c32');
\draw (c33)--(a33)--(b33)--(c33')--(c3');
\draw (b31)--(c32')--(b33);
\draw (a31)--(a32)--(a33)--(a31);

\draw (c12)[bend angle=30, bend right]to (c23); 
\draw (c13)[bend angle=28, bend left]to (c22); \draw (c22)[bend angle=30, bend left]to (c31);
\draw (c13)[bend angle=25, bend right]to (c31);

\end{tikzpicture} 
\caption{The perfect matching cut $(X,Y)$ of the example graph $G$ in Fig.~\ref{fig:P14freeExample1} given the assignment $\allowbreak y=v= \text{True}$, $x=z = u=w= \text{False}$.  
$X$ and $Y$ consist of the flax and teal vertices, respectively.}\label{fig:1in3->pmc}
\end{center}
\end{figure}

We now are ready to prove Theorem~\ref{thm:promise}: First note that by Lemmas~\ref{lem:P14free} and~\ref{lem:mc->pmc}, $G$ is  $\{3P_6,2P_7,P_{14}\}$-free $8$-chordal and every matching cut of~$G$ (if any) is a perfect matching cut. In particular, every matching cut of~$G$ is extendable to a perfect matching. 

Now, suppose $\phi$ has a 1-in-3 assignment. Then, by Lemma~\ref{lem:1in3->pmc}, $G$ has a perfect matching cut. In particular, $G$ has a disconnected perfect matching and, actually, a matching cut.

Conversely, let $G$ have a matching cut $M=(X,Y)$, possibly a perfect matching cut or one that is contained in a perfect matching of~$G$. Then, by Lemma~\ref{lem:mc-in-G}~(i), we may assume that $F\subseteq X$, and set variable $x$ to True if the corresponding variable clique~$Q(x)$ is contained in $Y$ and False if $Q(x)$ is contained in $X$. By Lemma~\ref{lem:mc-in-G}~(iii), this assignment is well defined. Moreover, it is a 1-in-3 assignment for $\phi$: consider a clause $C_j=\{x,y,z\}$ with $c_{j1}=x, c_{j2}=y$ and $c_{j3}=z$. By Lemma~\ref{lem:mc-in-G}~(ii) and~(iii), exactly one of $Q(x), Q(y)$ and~$Q(z)$ is contained in~$Y$, hence exactly one of $x,y$ and $z$ is assigned True.  

Finally, note that $G$ has $N=14m$ vertices and recall that, assuming ETH, \P1IN3SAT\ cannot be solved in $2^{o(m)}$ time. Thus, the ETH implies that no algorithm with runtime $2^{o(N)}$ exists for \textsc{promise-pmc mc}, \textsc{promise-pmc pmc} and \textsc{promise-pmc dpm}, even when restricted to $N$-vertex $\{3P_6,2P_7,P_{14}\}$-free $8$-chordal graphs. 

The proof of Theorem~\ref{thm:promise} is complete.\qed

\paragraph{Remark}  
It can easily be verified that the constructed graph $G$ has diameter~$4$ and radius~$3$. There exist a number of results on algorithms and complexity of matching cut problems restricted to graphs of bounded diameter and radius; the paper~\cite{LuckePR24} is an excellent source for this topic. 
While complexity dichotomies for \MC\ on graphs of bounded diameter and radius (\cite{LeL19,LuckePR23}), for \PMC\ on graphs of bounded radius (\cite{HeggernesT98,LuckePR23}), and for \DPM\ on graphs of bounded diameter (\cite{BouquetP25}) are known, it remain only two open cases in this research direction (see~\cite{LuckePR24}): what is the complexity of \PMC\ on graphs of diameter~3, and what is the complexity of \DPM\ on graphs of radius~2?

\section{Polynomial results: Proof of Theorem~\ref{thm:longholefree-dpm}}\label{sec:longholefree-dpm}
In this section, we prove Theorem~\ref{thm:longholefree-dpm}: 
\emph{\DPM\ and \PMC\ can be solved in polynomial time when restricted to $4$-chordal graphs.}

Recall that by Theorem~\ref{thm:longholefree-mc}, 
\MC\ is polynomially solvable for $4$-chordal graphs (also called \emph{quadrangulated graphs}). 
We first point out that \DPM\ is polynomially solvable for $4$-chordal graphs, too, by following known approaches~\cite{KratschL16,LeL19,Moshi89}. We note that the approach in~\cite{KratschL16,LeL19} has some similarity to Moshi's algorithm~\cite{Moshi89} for \MC\ on $4$-chordal graphs. However, it is not immediately clear how to modify the latter for solving \DPM\ restricted to $4$-chordal graphs.  We note also that both approaches for \MC\ and \DPM\ do not work for \PMC. 
Then, we will show that \PMC\ can be solved in polynomial time on $4$-chordal graphs by a new approach involving a non-trivial reduction to \textsc{2sat}.

\subsection{\DPM\ in $4$-chordal graphs}
Given a connected graph $G=(V,E)$ and two disjoint, non-empty vertex sets $A, B\subset V$ such that each vertex in $A$ is adjacent to exactly one vertex of $B$ and each vertex in~$B$ is adjacent to exactly one vertex of $A$. We say a matching cut of~$G$ is an \emph{$A,B$-matching cut} (or a matching cut \emph{separating $A$, $B$}) if $A$ is contained in one side and~$B$ is contained in the other side of the matching cut. Observe that $G$ has a matching cut if and only if~$G$ has an $\{a\},\{b\}$-matching cut for some edge $ab$, and $G$ has a disconnected perfect matching if and only if~$G$ has a perfect matching containing an $\{a\},\{b\}$-matching cut for some edge $ab$.

For each edge $ab$ of a $4$-chordal graph~$G$, we will be able to decide if $G$ has a disconnected perfect matching containing a matching cut separating $A=\{a\}$ and $B=\{b\}$. 
This is done by applying known forcing rules (\cite{KratschL16,LeL19,LeT22}), which are given below. 
Initially, set $X:=A$, $Y:=B$ and write $F=V(G)\setminus (X\cup Y)$ for the set of \lq free\rq\ vertices. 
The sets $A,B,X$ and~$Y$ will be extended, if possible, by adding vertices from~$F$ so that $A\subseteq X$, $B\subseteq Y$ and if $(U,W)$ is any $A$,$B$-matching cut of $G$ then either $X\subseteq U$ and $Y\subseteq Y$, or else $X\subseteq W$ and $Y\subseteq X$. 
In doing so, we will apply the following forcing rules exhaustively. The first three rules will detect certain vertices that ensure that $G$ cannot have an $A,B$-matching cut.  
\begin{itemize}
\item[(R1)] Let $v\in F$ be adjacent to a vertex in $A$. If $v$ is
   \begin{itemize}
      \item adjacent to a vertex in $B$, or 
      \item adjacent to (at least) two vertices in $Y\setminus B$,
   \end{itemize}
   then $G$ has no $A,B$-matching cut.
   
\item[(R2)] Let $v\in F$ be adjacent to a vertex in $B$. If $v$ is
   \begin{itemize}
      \item adjacent to a vertex in $A$, or 
      \item adjacent to (at least) two vertices in $X\setminus A$,
   \end{itemize}
   then $G$ has no $A,B$-matching cut.

\item[(R3)] If $v\in F$ is adjacent to (at least) two vertices in $X\setminus A$ and to (at least) two vertices in $Y\setminus B$, then $G$ has no $A,B$-matching cut.
\end{itemize}

The correctness of (R1), (R2) and (R3) is quite obvious. We assume that, before each application of the rules (R4) and (R5) below, none of (R1), (R2) and (R3) is applicable.

\begin{itemize}
\item[(R4)] Let $v\in F$ be adjacent to a vertex in $A$ or to (at least) two vertices in $X\setminus A$. Then $X:=X\cup\{v\}$, $F:=F\setminus\{v\}$. If, moreover, $v$ has a unique neighbor $w\in Y\setminus B$ then $A:=A\cup\{v\}$, $B:=B\cup\{w\}$.

\item[(R5)] Let $v\in F$ be adjacent to a vertex in $B$ or to (at least) two vertices in $Y\setminus B$. Then $Y:=Y\cup\{v\}$, $F:=F\setminus\{v\}$. If, moreover, $v$ has a unique neighbor $w\in X\setminus A$ then $B:=B\cup\{v\}$, $A:=A\cup\{w\}$. 
\end{itemize}

\smallskip
We refer to~\cite{LeL19} for the correctness of rules (R4) and (R5), and for the following facts. 

\begin{fact}\label{fact:a}
The total runtime for applying (R1) -- (R5) until none of the rules is applicable is bounded by $O(nm)$. 
\end{fact}

\begin{fact}\label{fact:b}
Suppose none of (R1) -- (R5) is applicable. Then 
\begin{itemize}
  \item $(X,Y)$ is an $A,B$-matching cut of $G[X\cup Y]$, and any $A,B$-matching cut of~$G$ must contain $X$ in one side and $Y$ in the other side;

  \item for any vertex $v\in F$, \[N(v)\cap A=\emptyset,\, N(v)\cap B=\emptyset \text{ and } |N(v)\cap X|\le 1,\, |N(v)\cap Y|\le 1.\]
\end{itemize}
\end{fact}

We now are ready to prove Theorem~\ref{thm:longholefree-dpm}. 
Let $G$ be a connected, $4$-chordal graph, and let $ab$ be an edge of $G$. Set $A=\{a\}$ and $B=\{b\}$, and assume that none of (R1) -- (R5) is applicable. 
For a connected component $S$ of $G[F]$, let $N(S)$ be the set of vertices outside $S$ adjacent to some vertex in $S$. 
Then, for any connected component $S$ of $G[F]$, 
\begin{align*}
  &\text{$N(S)\cap X=\emptyset$\, or\, $N(S)\cap Y=\emptyset$.}
\end{align*}
For, otherwise choose two vertices $s,s'\in S$ with a neighbor $x\in N(s)\cap X$ and a neighbor $y\in N(s')\cap Y$ such that the distance between $s$ and $s'$ in $S$ is as small as possible, 
as well as the distance between $x$ and $y$ in $G[X\cup Y]$ is as small as possible. Note that, by the choice of $s$ and $s'$, if $s$ and $s'$ are both adjacent to $x$ or $y$ then $s=s'$. Note also that, by the definition of~$X$ and~$Y$, $G[X\cup Y]$ is connected, and the distance between $x$ and $y$ in $G[X\cup Y]$ is at least three. Then $s, s', x$ and $y$ and a shortest $s,s'$-path in $S$, a shortest $x,y$-path in $G[X\cup Y]$ together would induce a long hole in~$G$. 

Partition $F$ into disjoint subsets $F_X$ and $F_Y$ as follows:
\begin{align*}
F_X& = \bigcup \{S : \text{$S$ is a component of $G[F]$ with $N(S)\cap X\not=\emptyset$}\},\\
F_Y& = \bigcup \{T : \text{$T$ is a component of $G[F]$ with $N(T)\cap Y\not=\emptyset$}\}.
\end{align*}
Then, by the facts above and recall that $G$ is connected, 
\begin{align*}
& F=F_X\cup F_Y \text{ and } F_X\cap F_Y=\emptyset.
\end{align*}
Thus, 
\begin{align*}
&\text{$(X\cup F_X, Y\cup F_Y)$ is an $A,B$- matching cut of $G$}, 
\end{align*}
and it follows, that 
\begin{align*}
&\text{$G$ has a disconnected perfect matching containing an $A,B$-matching}\\ 
&\text{cut if and only if $G-A-B$ has a perfect matching.}  
\end{align*}
Therefore, with Fact~\ref{fact:a}, in $O(nm)$ time we can decide whether~$G$ has a matching cut containing a given edge.\footnote{Thus, \MC\ is solvable in $O(nm^2)$ time when restricted to $4$-chordal graphs, an alternative proof of Theorem~\ref{thm:longholefree-mc}. We remark that Moshi's algorithm solving \MC\ in $4$-chordal graphs runs in $O(n^3m)$ time which is worse than ours for sparse graphs.} 
Moreover, as a maximum matching can be computed in $O(\sqrt{n}m)$ time~\cite{MicaliV80,Vazirani24}, we can decide in  $O(n\sqrt{n}m^2)$ time whether~$G$ has a disconnected perfect matching containing an $\{a\},\{b\}$-matching cut for a given edge $ab$.  Since there are at most $m$ edges to check, we can solve \DPM\ in $O(n\sqrt{n}m^3)$ time on $4$-chordal graphs. 

We note that the above approach cannot be adapted to solve \PMC\ because the last fact is not longer true for perfect matching cuts. For instance, if $G-A-B$ has a \emph{perfect matching cut} then $G$ does not need to have a \emph{perfect matching cut} containing an $A,B$-matching cut. 
In the next subsection, we will reduce \PMC\ restricted to $4$-chordal graphs to \textsc{2sat} instead.  

\subsection{\PMC\ in $4$-chordal graphs}
In this subsection, we present a polynomial-time algorithm solving \PMC\ in $4$-chordal graphs. 
Recall that, given a perfect matching cut $(X,Y)$ of a graph, every vertex $v\in X$ (respectively $v\in Y$) has exactly one neighbor $p(v)\in Y$ (respectively $p(v)\in X$) which is the private neighbor of $v$ with respect to $(X,Y)$. 
Our algorithm first partitions the vertex set of the input graph~$G$ into breadth-first-search levels, and then proceeds bottom-up from the last level to the root to create an equivalent \textsc{2-cnf} formula~$\phi$: $G$ has a perfect matching cut if and only if $\phi$ is satisfiable. 
In doing so, in each level, the algorithm will construct certain \textsc{2-cnf} clauses encoding the private neighbors of vertices in that level. If this is not possible, the algorithm will correctly decide that~$G$ has no perfect matching cut. Otherwise, the resulting \textsc{2sat}-instance $\phi$ is satisfiable if and only if~$G$ has a perfect matching cut.
 
Given a connected $4$-chordal graph $G=(V,E)$, we fix an arbitrary vertex~$r$ and partition~$V$ into \emph{levels} $L_i$ of vertices at distance $i$ from $r$,  
\[L_i=\{v\in V\mid \text{dist}(r,v)=i\}, i=0, 1, 2, \ldots\]
Let $h$ be the largest integer such that $L_h\not=\emptyset$ but $L_{h+1}=\emptyset$. 
Note that $L_0=\{r\}$ and that all level sets $L_i$, $0\le i\le h$, can be computed in linear time by running a breadth-first-search at $r$. To avoid triviality, we assume that $h\ge 2$. 
(If $h=1$ then $r$ is a universal vertex. Graphs with at least three vertices and a universal vertex have no perfect matching cut.) 

The following properties, Lemmas~\ref{lem:levels1} and~\ref{lem:levels2} below, of the level partition will be useful.
\begin{lemma}\label{lem:levels1} 
For any $2\le i\le h$ and any vertex $v\in L_i$, 
%
all vertices of any non-empty independent set in $N(v)\cap L_{i-1}$ have a common neighbor in $L_{i-2}$. 
\end{lemma}
\begin{proof}
We prove a slightly stronger statement: for any vertex $v\in L_i$, every two non-adjacent vertices $u_1, u_2\in N(v)\cap L_{i-1}$ satisfy $N(u_1)\cap L_{i-2}\subseteq N(u_2)\cap L_{i-2}$ or $N(u_2)\cap L_{i-2}\subseteq N(u_1)\cap L_{i-2}$. 

Observe that the statement is clearly true in case $i=2$, hence let $i>2$.  
Suppose the contrary that there are vertices $w_1\in \big(N(u_1)\setminus N(u_2)\big)\cap L_{i-2}$ and $w_2\in \big(N(u_2)\setminus N(u_1)\big)\cap L_{i-2}$. Then, if $w_1w_2$ is an edge of $G$ then $vu_1w_1w_2u_2v$ is an induced $5$-cycle in~$G$. If $w_1w_2$ is not an edge of $G$ then a shortest path in $G[L_0\cup L_1\cup\cdots\cup L_{i-3}\cup\{w_1,w_2\}]$ connecting $w_1$ and $w_2$ together with $v,u_1$ and $u_2$ form an induced cycle of length at least~$6$ in $G$. In any of these two cases, we reach a contradiction because $G$ is $4$-chordal. Thus,  every two non-adjacent vertices $u_1, u_2\in N(v)\cap L_{i-1}$ satisfy $N(u_1)\cap L_{i-2}\subseteq N(u_2)\cap L_{i-2}$ or $N(u_2)\cap L_{i-2}\subseteq N(u_1)\cap L_{i-2}$. 

Now, let $S\subseteq N(v)\cap L_{i-1}$ be an independent set and consider a vertex $w\in L_{i-2}$ with maximum $|N(w)\cap S|$. 
Then $N(w)\cap S=S$, that is,~$w$ is adjacent to all vertices in~$S$. For otherwise, let $u_1\in S\setminus N(w)$ and let $w'\in L_{i-2}$ be a neighbor of~$u_1$. By the choice of $w$, $w'$ is non-adjacent to a vertex $u_2\in N(w)\cap S$. But then the two vertices~$u_1$ and~$u_2$ contradict the statement above. 
\qed
\end{proof}

Recall that a vertex set $S\subseteq V$ is monochromatic if, for every matching cut $(X,Y)$ of~$G$,~$S$ is contained in $X$ or else in $Y$. Observe that if $G[S]$ is connected then~$S$ is monochromatic if and only if each edge $\{x,y\}$ of $G[S]$ is monochromatic. 
\begin{lemma}\label{lem:levels2} 
For any $1\le i\le h$, every connected component of $G[L_i]$ is monochromatic. In particular, for each vertex $v\in L_i$, $N[v]\cap L_i$ is monochromatic.
\end{lemma}
\begin{proof}
We proceed by induction on $i$. Let $A$ be a connected component in $G[L_1]$ with at least two vertices. Since all vertices in~$A$ are adjacent to the root~$r$, every edge in~$A$ forms with~$r$ a triangle. Hence each edge in~$A$ is monochromatic, and the statement is true for $i=1$.

Let us assume that the statement is true for $i<h$, and consider a connected component $A$ of $G[L_{i+1}]$. Let $\{x,y\}$ be an edge of $A$ which is not contained in a triangle (otherwise $\{x,y\}$ is monochromatic and we are done), and let $x'$, respectively, $y'$ be a neighbor in $L_i$ of $x$, respectively, $y$. Then $x'$ is not adjacent to~$y$ and~$y'$ is not adjacent to~$x$; in particular, $x'\not=y'$. Now, $x'$ and $y'$ must be adjacent, otherwise a shortest path in $G[L_0\cup L_1\cup\cdots\cup L_{i-1}\cup\{x',y'\}]$ connecting~$x'$ and~$y'$ together with~$x$ and~$y$ would form an induced cycle of length at least~$5$ in~$G$.  
Now, by induction hypothesis, $\{x',y'\}$ is monochromatic, implying $\{x,y\}$ is monochromatic, too. 
(Otherwise, $G$ has no (perfect) matching cut; $x$ or $y$ would have two neighbors in the other part.) 
\qed
\end{proof}

Now, if $G$ has a perfect matching cut then, by Lemma~\ref{lem:levels2}, every vertex $v\in L_i$ must have its private neighbor $p(v)$ in $L_{i-1}$ or in $L_{i+1}$. 
In particular, every vertex in $L_h$ must have its private neighbor in $L_{h-1}$, and as we will see, its private neighbor can correctly be determined or encoded by at most four \textsc{2-cnf} clauses. Thus, after all vertices in the last level $L_h$ are assigned to their private neighbors in $L_{h-1}$, the unassigned vertices in $L_{h-1}$, if any, must have their private neighbors in $L_{h-2}$. 
So, we will proceed from the bottom up to assign every vertex to its private neighbor by creating a \textsc{2-cnf} formula or, in case some vertex cannot be assigned, report that~$G$ does not have any perfect matching cut. 

To this end, we define a vertex subset $Q\subseteq V(G)$ of a graph $G$ to be \emph{determined} if, for every perfect matching cut of $G$, every vertex $v\in Q$ has its private neighbor $p(v)\in Q$. 
Vertices in $Q$ are determined, all other are undetermined. Note that if $G$ has a perfect matching cut then every vertex $v\not\in Q$ has \emph{non-monochromatic} $N[v]\setminus Q$ (as the private neighbor of $v$ exists and is outside $Q$). 
Then, we call a vertex $v\in L_i$ a \emph{leaf} if $(N(v)\cap L_{i+1}\setminus Q)\cup\{v\}$ is monochromatic. 
Note that, since $L_{h+1}=\emptyset$, all vertices in the last level $L_h$ are leaves. 

The correctness of our approach is based on the following crucial fact.
\begin{lemma}\label{lem:leaf}
Let $G=(V,E)$ be a $4$-chordal graph and let $Q\subseteq V$ be a determined set. 
Assuming $G$ has a perfect matching cut, any leaf $v\in L_i\setminus Q$, $i>0$, satisfies exactly one of the following conditions.
\begin{itemize}
\item[\em (c1)] $v$ has exactly one neighbor, $u$, in $L_{i-1}\setminus Q$. 
In this case, the private neighbor of $v$ is $u$; 
\item[\em (c2)] $v$ has exactly two neighbors, $u_1$ and $u_2$, in $L_{i-1}\setminus Q$, and 
there is a vertex $w\in L_{i-2}\setminus Q$ adjacent to both $u_1$ and $u_2$. 
In this case, one of $u_1,u_2$ is the private neighbor of $v$ and the other is the private neighbor of $w$; 
\item[\em (c3)] $v$ has more than two neighbors in $L_{i-1}\setminus Q$, one in a component $A$ and all others in another component $B$ of $G[L_{i-1}\setminus Q]$. 
In this case, the private neighbor of $v$ is its neighbor in $A$. 
\end{itemize}
\end{lemma}
\begin{proof}
Observe first that any leaf $v\in L_i\setminus Q$ has a neighbor $u\in L_{i-1}\setminus Q$ since otherwise $N[v]\setminus Q$ would be monochromatic. 
Furthermore, if a leaf $v\in L_i\setminus Q$ has exactly one neighbor,~$u$, in $L_{i-1}\setminus Q$ then, by Lemma~\ref{lem:levels2}, $u$ must be the private neighbor of $v$ in any perfect matching cut of~$G$. 
In particular, in case $i=1$ condition (c1) clearly holds. So, let~$i>1$.  

Next, assume that $N(v)\cap L_{i-1}\setminus Q$ is contained in $k\ge 3$ connected components $A_1, A_2, \ldots A_k$ of $G[L_{i-1}\setminus Q]$, and let $u_j\in N(v)\cap A_j$, $1\le j\le k$. 
By Lemma~\ref{lem:levels1}, there is a vertex $w\in L_{i-2}$ adjacent to all $u_j$'s, and so $\{v,w,u_1,u_2,\ldots, u_k\}$ induces a $K_{2,k}$ which is monochromatic. 
Note that each $A_j$ is contained in some connected component of $G[L_{i-1}]$. Hence, by Lemma~\ref{lem:levels2}, it follows that $N[v]\setminus Q$ is monochromatic, a contradiction. 
Thus, $N(v)\cap L_{i-1}\setminus Q$ is contained in at most two connected components of $G[L_{i-1}\setminus Q]$. 

If $N(v)\cap L_{i-1}\setminus Q$ is contained in exactly one connected component of $G[L_{i-1}\setminus Q]$, say $A$, then 
$|N(v)\cap A|=1$, that is, (c1) holds. This because if $|N(v)\cap A|\ge 2$ then, by Lemma~\ref{lem:levels2} again, $A\cup \{v\}$ would be monochromatic, and therefore $N[v]\setminus Q$ would be monochromatic. 

So, let us assume that $N(v)\cap L_{i-1}\setminus Q$ is contained in exactly two connected components of $G[L_{i-1}\setminus Q]$, say~$A$ and~$B$. 
Then $|N(v)\cap A|=1$ or $|N(v)\cap B|=1$, otherwise by Lemma~\ref{lem:levels2} again, $N[v]\setminus Q$ would be monochromatic. 

If $|N(v)\cap B|\ge 2$, say, then the private neighbor of $v$ must be its neighbor in $A$ because $N[v]\setminus A$ is monochromatic, and $v$ therefore satisfies condition (c3).   

It remains to consider the case $|N(v)\cap A|=1$ and $|N(v)\cap B|=1$.  Let $N(v)\cap A=\{u_1\}$ and $N(v)\cap B=\{u_2\}$. By Lemma~\ref{lem:levels1}, there is a vertex $w\in L_{i-2}$ adjacent to $u_1$ and $u_2$, and thus $C: vu_1u_2wv$ is an induced $4$-cycle. 
Note that $w$ is undetermined: were $w\in Q$ (so its private neighbor $p(w)\in Q$), were $N[w]\setminus\{p(w)\}$, hence $N[v]\setminus Q$,  monochromatic. Thus, $w\not\in Q$. 
Now, $u_1$ and~$u_2$ must belong to different parts of any perfect matching cut of $G$, since otherwise $\{v,u_1,u_2\}$, hence $N[v]\setminus Q$, would be monochromatic. This implies that one of $u_1, u_2$ is the private neighbor of $v$ and the other is the private neighbor of $w$. Thus, $v$ satisfies condition (c2). 
\qed
\end{proof}

Now, to decide whether $G$ admits a perfect matching cut we proceed as follows. 
First, we create for each vertex $v$ of $G$ a Boolean variable, which is also denoted by~$v$ as the context will be clear. Then, at the beginning, all vertices are undetermined, $Q=\emptyset$, and we process bottom-up with a leaf~$v$ in $L_h$ by checking if~$v$ satisfies one of the conditions in Lemma~\ref{lem:leaf}. If not, $G$ has no perfect matching cut. Otherwise, we create certain \textsc{2-cnf} clauses that assign~$v$ to its private neighbor $p(v)$, and force~$v$ and~$x$, and $p(v)$ and~$y$ to be monochromatic for all other neighbors $x\not=p(v)$ of~$v$ and all other neighbors $y\not=v$ of $p(v)$. 
Then~$v$ and its private neighbor $p(v)$ become determined. 
This way we will correctly decide if $G$ has no perfect matching cut or we will successfully construct a \textsc{2sat}-instance $\phi$ such that~$G$ has a perfect matching cut if and only if~$\phi$ is satisfiable.

The details are given in Algorithm~\ref{algo:PMC}. At this point, we note that, for condition (c2) in Lemma~\ref{lem:leaf}, the fact that one of $u_1, u_2$ is the private neighbor of $v$ and the other is the private neighbor of $w$ can be equivalently stated that $\{u_1,u_2\}$ and $\{v,w\}$ are non-monochromatic. The latter can be expressed by four \textsc{2-cnf} clauses (lines~\ref{line:c2-private} in the algorithm). 

\begin{algorithm}[!ht]
\SetNlSty{tiny}{}{} 
\DontPrintSemicolon 
\KwIn{A connected $4$-chordal graph $G=(V,E)$}
\KwOut{Either a pmc $(X,Y)$ of $G$, or `NO' if $G$ has no pmc} 
\BlankLine

\SetKw{DownTo}{downto}

fix a vertex $r\in V$ and compute the BFS-levels $L_i$ from $r$, $0\le i\le h$\; 
create for each vertex $v\in V$ a Boolean variable $v$\;
$Q\leftarrow \emptyset$; $\phi\leftarrow \emptyset$\;

\For{$i\leftarrow h$ \DownTo $1$}{\label{line:for-loop}
   \ForEach{$v\in L_i\setminus Q$}{\label{line:foreach-v-Li} 
     \lIf{$v$ does not satisfy any of (c1), (c2) and (c3) in Lemma~\ref{lem:leaf}}{\label{line:NO-c1c2c3}\Return `NO'} 
     \If{$v$ satisfies (c1) or (c3) in Lemma~\ref{lem:leaf}}{\label{line:c1c3}
        let $A$ be the connected component of $G[L_{i-1}\setminus Q]$ s.t. $N(v)\cap A=\{u\}$\;
        $\phi\leftarrow \phi\cup\{(v\lor u),(\neg v\lor \neg u)\}$\;\label{line:c1c3-private} 
        $Q\leftarrow Q\cup\{v,u\}$\;\label{line:c1c3-assigned}
        \lForEach{$x\in N(v)\setminus Q$}{
         $\phi\leftarrow \phi\cup\{(v\lor \neg x),(\neg v\lor x)\}$\label{line:c1c3-others-v}
         }
        \lForEach{$x\in N(u)\setminus Q$}{
         $\phi\leftarrow \phi\cup\{(u\lor \neg x),(\neg u\lor x)\}$\label{line:c1c3-others-u}
         }
        }
     \If{$v$ satisfies (c2) in Lemma~\ref{lem:leaf}}{\label{line:c2}
         let $u_1, u_2$ be the two neighbors of $v$ in $L_{i-1}\setminus Q$, and 
         let $w\in L_{i-2}\setminus Q$ be such that $v,u_1,u_2,w$ induce a $4$-cycle\;
         $\phi\leftarrow \phi\cup\{(v\lor w),(\neg v\lor \neg w),(u_1\lor u_2),(\neg u_1\lor \neg u_2)\}$\;\label{line:c2-private}
         $Q\leftarrow Q\cup\{v,u_1,u_2,w\}$\;\label{line:c2-assigned}
        \lForEach{$x\in N(v)\setminus Q$}{
         $\phi\leftarrow \phi\cup\{(v\lor \neg x),(\neg v\lor x)\}$\label{line:c2-others-v}
         }
        \lForEach{$x\in N(w)\setminus Q$}{
         $\phi\leftarrow \phi\cup\{(w\lor \neg x),(\neg w\lor x)\}$\label{line:c2-others-w}
         }
        \lForEach{$x\in N(u_1)\setminus Q$}{
         $\phi\leftarrow \phi\cup\{(u_1\lor \neg x),(\neg u_1\lor x)\}$\label{line:c2-others-u1}
         }
        \lForEach{$x\in N(u_2)\setminus Q$}{
         $\phi\leftarrow \phi\cup\{(u_2\lor \neg x),(\neg u_2\lor x)\}$\label{line:c2-others-u2}
         }
      }
   }
}
\label{line:end-for}
\label{line:2sat}\lIf{$\phi$ is satisfiable}{\Return $(X,Y)$ with $X$ the true and $Y$ the false vertices} 
\lElse{\label{line:NO-2sat}\Return `NO'}

\caption{\texttt{Recognizing~$4$-chordal graphs having pmc}} 
\label{algo:PMC}
\end{algorithm} 

\paragraph*{Two small examples} 
In Fig.~\ref{fig:two-examples}, consider the graph on the left. Letting $r=v_2$, we have $L_0=\{v_2\}$, $L_1=\{v_0,v_1,v_3\}$ and $L_2=\{v_4,v_5\}$.  Depending on the choice of $v\in L_2\setminus Q$ in line~\ref{line:foreach-v-Li} of the algorithm, we have two cases: 

\begin{figure}[H]
\begin{center}
\tikzstyle{vertexS}=[draw,circle,inner sep=2pt,fill=black] 
\tikzstyle{vertex}=[draw,circle,inner sep=1.3pt,fill=black] 
\begin{tikzpicture}[scale=.36] 
\node[vertex] (a) at (1,1)  [label=below:{\small $v_1$}]{};
\node[vertex] (b) at (1,4)  [label=above:{\small $v_0$}]{};
\node[vertex] (c) at (7,4)  [label=above:{\small $v_4$}]{}; 
\node[vertex] (d) at (7,1)  [label=below:{\small $v_5$}]{};
\node[vertex] (e) at (3,2.5)  [label=above:{\small $v_2$}]{};
\node[vertex] (f) at (5,2.5)  [label=above:{\small $v_3$}]{};

\draw (a)--(b)--(e); \draw (e)--(f); \draw (f)--(c)--(d); \draw (d)--(a); \draw (a)--(e);
\draw (f)--(d); 
\end{tikzpicture} 
\qquad
\begin{tikzpicture}[scale=.36] 
\node[vertex] (a) at (1,1)  [label=below:{\small $v_2$}]{};
\node[vertex] (b) at (1,4)  [label=above:{\small $v_0$}]{};
\node[vertex] (c) at (4,4)  [label=above:{\small $v_1$}]{};
\node[vertex] (d) at (4,1)  [label=below:{\small $v_3$}]{};
\node[vertex] (e) at (7,1)  [label=below:{\small $v_5$}]{};
\node[vertex] (f) at (7,4)   [label=above:{\small $v_4$}]{};

\draw (a)--(b); \draw (b)--(c)--(d)--(a);
\draw (d)--(e); \draw (e)--(f); \draw (c)--(f); 
\end{tikzpicture} 
\caption{Two small examples to Algorithm~\ref{algo:PMC}.}\label{fig:two-examples}
\end{center}
\end{figure}
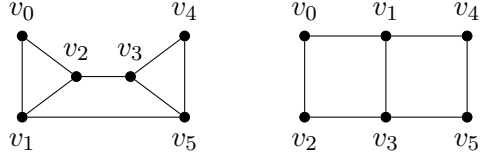

\begin{itemize}
\item Suppose the algorithm first picks $v=v_4$. Then the formula $\phi$ consists of the following clauses: 
$(v_4\lor v_3)$, $(\neg v_4\lor \neg v_3)$ (line~\ref{line:c1c3-private}), $(v_4\lor \neg v_5)$, $(\neg v_4\lor v_5)$, $(v_3\lor \neg v_5)$, $(\neg v_3\lor v_5)$, $(v_3\lor \neg v_2)$, $(\neg v_3\lor v_2)$ (lines~\ref{line:c1c3-others-v}, \ref{line:c1c3-others-u}), 
$(v_5\lor v_1)$, $(\neg v_5\lor \neg v_1)$ (line~\ref{line:c1c3-private}), $(v_1\lor \neg v_0)$, $(\neg v_1\lor v_0)$, $(v_1\lor \neg v_2)$, $(\neg v_1\lor v_2)$ (lines~\ref{line:c1c3-others-v}, \ref{line:c1c3-others-u}), and 
$(v_0\lor v_2)$, $(\neg v_0\lor \neg v_2)$ (line~\ref{line:c1c3-private}). In this case, the output is \lq NO\rq\ as $\phi$ is unsatisfiable (line~\ref{line:NO-2sat}).
\item Suppose the algorithm first picks $v=v_5$. Then the formula $\phi$ consists of the following clauses: 
$(v_5\lor v_2)$, $(\neg v_5\lor \neg v_2)$, $(v_1\lor v_3)$, $(\neg v_1\lor \neg v_3)$ (line~\ref{line:c2-private}), $(v_5\lor \neg v_4)$, $(\neg v_5\lor v_4)$, $(v_2\lor \neg v_0)$, $(\neg v_2\lor v_0)$, $(v_1\lor \neg v_0)$, $(\neg v_1\lor v_0)$, $(v_3\lor \neg v_4)$, $(\neg v_3\lor v_4)$ (lines~\ref{line:c2-others-v}--\ref{line:c2-others-u2}). 
In this case, the output is \lq NO\rq\ as in the next round $v=v_4\in L_2\setminus Q$ does not satisfy any of the conditions (c1), (c2) and (c3) (line~\ref{line:NO-c1c2c3}).
\end{itemize}

Let us consider the second graph in Fig.~\ref{fig:two-examples}, on the right. We want to choose two different roots for this example.
\begin{itemize}
\item 
First let $r=v_0$, implying $L_0=\{v_0\}$, $L_1=\{v_1,v_2\}$, $L_2=\{v_3,v_4\}$ and $L_3=\{v_5\}$. Then the formula $\phi$ consists of the following clauses: 
$(v_5\lor v_1)$, $(\neg v_5\lor \neg v_1)$, $(v_3\lor v_4)$, $(\neg v_3\lor \neg v_4)$ (line~\ref{line:c2-private}), $(v_1\lor \neg v_0)$, $(\neg v_1\lor v_0)$, $(v_3\lor \neg v_2)$, $(\neg v_3\lor v_2)$ (lines~\ref{line:c2-others-v}--\ref{line:c2-others-u2}), and 
$(v_2\lor v_0)$, $(\neg v_2\lor \neg v_0)$ (line~\ref{line:c1c3-private}). In this case, $\phi$ has a satisfying assignment, say $v_0=v_1=v_4= \text{True}$ and $v_2=v_3=v_5=\text{False}$, and the output is the perfect matching cut $(X,Y)$ with $X=\{v_0,v_1,v_4\}$ and $Y=\{v_2,v_3,v_5\}$ (line~\ref{line:2sat}). 
\item 
Next let $r=v_1$, implying $L_0=\{v_1\}$, $L_1=\{v_0,v_3,v_4\}$ and $L_2=\{v_2,v_5\}$. 
By symmetry, we may assume that the algorithm first picks $v=v_2\in L_2\setminus Q$ at line~\ref{line:foreach-v-Li}. Then the formula $\phi$ consists of the following clauses: 
$(v_2\lor v_1)$, $(\neg v_2\lor \neg v_1)$, $(v_0\lor v_3)$, $(\neg v_0\lor \neg v_3)$ (line~\ref{line:c2-private}), $(v_1\lor \neg v_4)$, $(\neg v_1\lor v_4)$, $(v_3\lor \neg v_5)$, $(\neg v_3\lor v_5)$ (lines~\ref{line:c2-others-v}--\ref{line:c2-others-u2}), and 
$(v_5\lor v_4)$, $(\neg v_5\lor \neg v_4)$ (line~\ref{line:c1c3-private}). Again, $\phi$ has a satisfying assignment, say $v_0=v_1=v_4= \text{True}$ and $v_2=v_3=v_5=\text{False}$, and the output is the perfect matching cut $(X,Y)$ with $X=\{v_0,v_1,v_4\}$ and $Y=\{v_2,v_3,v_5\}$ (line~\ref{line:2sat}). 
\end{itemize}
\paragraph*{Time complexity} 
Observe that $\phi$ has $O(|E|)$ clauses. 
Since checking the conditions (c1), (c2) and (c3) in Lemma~\ref{lem:leaf}, as well as deciding if a \textsc{2sat} formula is satisfiable and if so computing a satisfying assignment~\cite{AspvallPT79}, can be done in polynomial time, the running time of Algorithm~\ref{algo:PMC} clearly is polynomial. 

\paragraph*{Correctness}
For the correctness, we first show:
\begin{fact}\label{fact:correctness-for}
After any round $i$ of the for-loop, lines~\ref{line:for-loop}--\ref{line:end-for}, Algorithm~\ref{algo:PMC} either correctly decided that $G$ does not have a perfect matching cut, or correctly assigned all vertices in $L_i$ to their private neighbors:  $L_{i}\subseteq Q$.
\end{fact}
\begin{proof}
We prove this fact by induction on $i$. First consider the case $i=h$. In this case, any $v\in L_h\setminus Q$ is a leaf (as $L_{h+1}=\emptyset$). Hence, by Lemma~\ref{lem:leaf}, every $v\in L_h\setminus Q$ must satisfy (c1), (c2) or (c3), otherwise~$G$ has no perfect matching cuts and the output NO at line~\ref{line:NO-c1c2c3} is correct. Thus, suppose that $G$ has a perfect matching cut.

If $v$ satisfies (c1) or (c3) with $A$ is the connected component of $G[L_{h-1}\setminus Q]$ containing exactly one neighbor of~$v$, say $u$, then $u$ must be the private neighbor of~$v$ (and $v$ therefore must be the private neighbor of $u$). Thus, in this case, the \textsc{2-cnf} clauses at line~\ref{line:c1c3-private} correctly assign~$v$ to $p(v)=u$ and~$u$ to $p(u)=v$, and subsequently $v$ and~$u$ become determined at line~\ref{line:c1c3-assigned}. 
Finally, all other neighbors~$x\not=u$ of~$v$ must belong to the same part with~$v$ in any perfect matching cut of $G$ (line~\ref{line:c1c3-others-v}), and all other neighbors~$x\not=v$ of~$u$ must belong to the same part with~$u$ (line~\ref{line:c1c3-others-u}). 

We now consider the case where $v$ satisfies (c2).  In this case, the \textsc{2-cnf} clauses at line~\ref{line:c2} correctly encode the fact that one of $u_1, u_2$ is the private neighbor of $v$ and the other is the private neighbor of $w$. Then the four vertices become determined (line~\ref{line:c2-assigned}).  
Finally, as in the above case, we have to force other neighbors of $v, w, u_1, u_2$ to belong to the same part with $v, w, u_1, u_2$, respectively (lines~\ref{line:c2-others-v}--\ref{line:c2-others-u2}). 
Thus the statement holds in case $i=h$. 

Let $i<h$. By induction, all vertices in $L_{i+1}$ are correctly assigned: $L_{i+1}\subseteq Q$. Thus, any vertex $v\in L_i\setminus Q$ is a leaf, and Lemma~\ref{lem:leaf} is applicable for all vertices in $L_i\setminus Q$. Then by the same arguments as in the case of round~$h$ above, it follows that the statement holds for round~$i$.
\qed
\end{proof}
By Fact~\ref{fact:correctness-for}, it remains to show that the output of Algorithm~\ref{algo:PMC} at line~\ref{line:2sat} is correct, i.e., $G$ has a perfect matching cut if and only if the constructed \textsc{2sat}-instance $\phi$ is satisfiable. 

Suppose first that $\phi$ admits a truth assignment. Then the partition $(X,Y)$ of $V(G)$ with~$X$ the set of all \lq true\rq\ vertices and $Y$ the set of all \lq false\rq\ vertices is a perfect matching cut of~$G$: the clauses constructed at lines~\ref{line:c1c3-private} and~\ref{line:c2-private} ensure that $X$ and $Y$ are non-empty, and any vertex in $X$ has a neighbor in $Y$ and vice versa. Moreover, the clauses constructed at lines~\ref{line:c1c3-others-v} and~\ref{line:c1c3-others-u}, and at lines~\ref{line:c2-others-v}--\ref{line:c2-others-u2} ensure that if $x\in X$ (respectively $x\in Y$) has a neighbor $p(x)\in Y$ (respectively $p(x)\in X$) then all neighbors $\not= p(x)$ of $x$ belong to the same part with $x$ and all neighbors $\not= x$ of $p(x)$ belong to the same part with $p(x)$. That is, $x$ has exactly one neighbor $p(x)$ in the other part. 

Conversely, suppose that $G$ has a perfect matching cut $(X,Y)$. Then set variable~$v$ to True if the corresponding vertex~$v$ is in~$X$, and False if it is in~$Y$. 
With Lemma~\ref{lem:leaf},  
we argue that~$\phi$ is satisfied by this truth assignment: the two clauses at line~\ref{line:c1c3-private} are satisfied because in this case $u$ is the private neighbor of $v$ and $v$ is the private neighbor of $u$, hence $u$ and $v$ must belong to different parts. 
Also, the clauses at lines~\ref{line:c1c3-others-v} and~\ref{line:c1c3-others-u} are satisfied because the other neighbors~$x$ of~$v$, respectively, of~$u$ must belong to the same part as~$v$, respectively, as~$u$. Similar, the four clauses at line~\ref{line:c2-private} are satisfied because in this case $v$ and $w$, as well as $u_1$ and $u_2$ must belong to different parts. Finally, the clauses at lines~\ref{line:c2-others-v}--\ref{line:c2-others-u2} are satisfied because the other neighbors~$x$ of $v,w,u_1$ and $u_2$ must belong to the same part as $v,w,u_1$ and~$u_2$, respectively.

\section{Conclusions}\label{sec:con}
We have shown that all three problems \textsc{matching cut}, \textsc{perfect matching cut} and \textsc{disconnected perfect matching} are $\NP$-complete in $P_{14}$-free $8$-chordal graphs. The hardness result for \textsc{perfect matching cut} in $P_{14}$-free graphs solves an open problem posed in~\cite{LuckePR23}. 
For \textsc{matching cut} and \textsc{disconnected perfect matching}, the hardness result improves the previous results in~\cite{LuckePR23} in $P_{15}$-free graphs, respectively, in $P_{19}$-free graphs, to $P_{14}$-free graphs. 
An obvious question is whether one of these problems remains $\NP$-complete in $P_t$-free graphs for some $t<14$. 

The hardness result for \textsc{perfect matching cut} in $8$-chordal graphs partly solves an open problem posed in~\cite{LeT22,LuckePR24}.  
We have proved that 
\textsc{disconnected perfect matching} and \textsc{perfect matching cut} are solvable in polynomial time when restricted to $4$-chordal graphs. 
Thus, another obvious question is whether one of these problems remains $\NP$-complete in $k$-chordal graphs for $5\le k\le7$.

\paragraph{Acknowledgments} We would like to thank one of the reviewers for her/his extremely careful reading and many detailed and useful comments.

\bibliographystyle{plainurl}
\bibliography{wg23JCSS}

\end{document}